\documentclass[11pt]{article}
\usepackage{snapshot}
\usepackage{paralist}
\usepackage{color}
\usepackage{bm}
\usepackage{mdframed}
\usepackage{soul}
\usepackage{tcolorbox}
\usepackage{graphicx}

\usepackage{url}
\usepackage{xspace}
\usepackage[mathscr]{euscript}

\usepackage[cm]{fullpage}

\usepackage{amsfonts}
\usepackage{amssymb}

\usepackage{amsmath}

\usepackage{constants}
\usepackage{footnote}
\makesavenoteenv{tabular}

\usepackage{amsthm}
\usepackage{longtable}

\makeatletter
\newtheorem*{rep@theorem}{\rep@title}
\newcommand{\newreptheorem}[2]{%
\newenvironment{rep#1}[1]{%
 \def\rep@title{#2 \ref{##1}}%
 \begin{rep@theorem}}%
 {\end{rep@theorem}}}
\makeatother

\newenvironment{tightenumerate}{
	\begin{enumerate}[topsep=0pt,itemsep=0ex,partopsep=0ex,parsep=-1ex,leftmargin=15pt]
		\setlength{\parskip}{0pt}
	}{\end{enumerate}}

\allowdisplaybreaks

\usepackage{ifpdf}
\newcommand{\mydriver}{hypertex}
\ifpdf
 \renewcommand{\mydriver}{pdftex}
\fi
\usepackage[breaklinks,\mydriver]{hyperref}

\theoremstyle{plain}
\newtheorem{theorem}{Theorem}[section]%
\newtheorem{fact}[theorem]{Fact}
\newtheorem{lemma}[theorem]{Lemma}

\newtheorem{claim}[theorem]{Claim}

\newtheorem{definition}[theorem]{Definition}

\theoremstyle{definition}

\usepackage[margin=1in]{geometry}
\usepackage{footnote}

\newcommand{\TV}{\textrm{TV}}
\newcommand{\vol}{\textrm{vol}}

\newcommand{\E}{\textrm{E}}

\newcommand{\poly}{\textrm{poly}}

\newcommand{\HH}{\ensuremath{\mathcal{H}}}
\newcommand{\MM}{\ensuremath{\mathcal{M}}}

\newcommand{\p}{\textbf{p}}
\newcommand{\R}{\mathbb{R}}
\newcommand{\pp}{{\textbf{a}}}
\newcommand{\PP}{\ensuremath{\mathcal{P}}}

\newcommand{\qq}{{\textbf{b}}}

\newcommand{\vv}{\textbf{v}}

\newcommand{\uni}{\ensuremath{\mathbf{\mathcal{U}}}}
\newcommand{\bb}{\textbf{b}}

\newcommand{\1}{\textbf{1}}

\newcommand{\D}{\textbf{D}}
\newcommand{\A}{\textbf{A}}
\newcommand{\rcp}{\ensuremath{\mathrm{rcp}}}
\newcommand{\EstRCP}{\textsc{EstimateRCP}}
\newcommand{\Span}{\textrm{span}}

\newcommand{\I}{\textbf{I}}

\renewcommand{\qed}{\nobreak \ifvmode \relax \else
	\ifdim\lastskip<1.5em \hskip-\lastskip
	\hskip1.5em plus0em minus0.5em \fi \nobreak
	\vrule height0.75em width0.5em depth0.25em\fi}

\providecommand{\abs}[1]{\lvert#1\rvert} \providecommand{\norm}[1]{\lVert#1\rVert}
\newconstantfamily{c}{symbol=c}
\newcommand{\cst}{\ensuremath{\kappa}}

\newcommand{\ccss}{\ensuremath{\gamma}}

\newconstantfamily{smallconst}{symbol=\kappa}

\newcommand{\myout}{\textrm{out}}
\newcommand{\myin}{\textrm{in}}

\newcommand{\Conn}{\textrm{Conn}}

\providecommand{\abs}[1]{\left|#1\right|}
\providecommand{\norm}[1]{\lVert#1\rVert}
\providecommand{\vect}[1]{\ensuremath\mathbf#1}
\providecommand{\agbracket}[1]{\langle#1\rangle}
\providecommand{\mat}[1]{\ensuremath\mathbf#1}

\usepackage{enumitem}

\title{Robust Clustering Oracle and Local Reconstructor of Cluster Structure of Graphs}
\usepackage{times}
\date{}

\author{
	Pan Peng\footnote{	Department of Computer Science, University of Sheffield, Sheffield, U.K. Email: \url{p.peng@sheffield.ac.uk}. %
	}
}

\begin{document}

\begin{titlepage}
		
\maketitle

\begin{abstract}
Due to the massive size of modern network data, local algorithms that run in \emph{sublinear} time for analyzing the cluster structure of the graph are receiving growing interest. Two typical examples are local graph clustering algorithms that find a cluster from a seed node with running time proportional to the size of the output set, and clusterability testing algorithms that decide if a graph can be partitioned into a few clusters in the framework of property testing. 
  
In this work, we develop sublinear time algorithms for analyzing the cluster structure of graphs with noisy partial information. By using \emph{conductance} based definitions for measuring the quality of clusters and the cluster structure, we formalize a definition of \emph{noisy clusterable graphs} with bounded maximum degree. The algorithm is given query access to the adjacency list to such a graph. We then formalize the notion of \emph{robust clustering oracle} for a noisy clusterable graph, and give an algorithm that builds such an oracle in sublinear time, which can be further used to support typical queries (e.g., \textsc{IsOutlier}($s$), \textsc{SameCluster}($s,t$)) regarding the cluster structure of the graph in sublinear time. All the answers are consistent with a partition of $G$ in which all but a small fraction of vertices belong to some good cluster. We also give a \emph{local reconstructor} for a noisy clusterable graph that provides query access to a reconstructed graph that is guaranteed to be clusterable in sublinear time. All the query answers are consistent with a clusterable graph which is guaranteed to be close to the original graph.

To obtain our results, we give new analysis of the behavior of random walks on a noisy clusterable graph, which consists of a large subset that induces a clusterable graph and a small unknown subgraph (the noise). We show that a random walk of appropriately chosen length from a typical vertex in a large cluster of the clusterable part will mix well in the corresponding cluster. Using this we are able to distinguish vertices from the clusterable part from those in the noisy part. 
\end{abstract}

\end{titlepage}

\section{Introduction}
Graph clustering is a fundamental task arising from many domains, including computer science, social science, network analysis and statistics. Given a graph, the task is to group the vertices into \emph{reasonably good} clusters, where vertices inside the same cluster are well-connected to each other, and any two different clusters are well-separated. Such clusters convey valuable information of large graphs, and have concrete applications in recommendation systems, search engine, network routing and many others~(see e.g., surveys~\cite{Sch07:graph,POM09:communities,For10:community,New12:communities}). %
Many efficient \emph{global} clustering algorithms that run in polynomial time have been proposed for analyzing the structure of graphs, where the goal is to find the overall cluster structure of a graph. Almost all such algorithms need to at least read the whole input of the graph and thus run in linear time. Actually, even just outputting all the clusters will require $\Omega(n)$ time, where $n$ is the number of the vertices of the graph. These algorithms, though considered to be efficient in the classical algorithm design, are becoming impractical (and sometimes even impossible) to be used for processing and analyzing modern very large networks/graphs (e.g., WWW and social networks). 

Therefore, local algorithms that run in \emph{sublinear} time for analyzing the cluster structure of the graph are receiving growing interest. Such algorithms are typically assumed to be able to explore the input graph by performing appropriate queries, e.g., query the degree or the neighbor of any node. %
There have been two main frameworks for designing sublinear algorithms for graph clustering, if one uses the well-motivated notion \emph{conductance} (see below) to measure the quality of clusters. In the first one, called \emph{local graph clustering}, the goal is to find a cluster from a specified vertex with running time that is bounded in terms of the size of the output set (and with a weak dependence on $n$)~(see e.g., \cite{ST13:local,ACL06:local,AP09:sparse,OT12:local,AOPT16:local,ZLM2013:local,OZ14:flow}). If the target cluster has much small size, then the running time of the resulting algorithm will be sublinear in the input size. In the second one, called \emph{testing cluster structure} in the framework of \emph{property testing}, the goal is to distinguish if an input graph has a typical cluster structure or is far from such cases~(see \cite{CPS15:cluster,CKKMP18:cluster} and more discussions below). Such algorithms make decisions on the global cluster structure of the input graph by sampling vertices and locally exploring a small portion of the graph, and they can be served as a preliminary step before learning the cluster structure.

In this work, we study local and sublinear algorithms for analyzing the cluster structure of graphs that may contain  noise and/or outliers. In many real applications, due to external noise or errors, the network data set may fail to have the desired property (here, the cluster structure), while it might still be close to have this property. That is, the graph $G$ under our consideration is some kind of \emph{perturbation of a clusterable graph} or a \emph{noisy clusterable graph}: $G$ is first chosen from some class of clusterable graphs with an underlying while unknown partition, and then some noise and/or outliers are introduced by some adversary or in some random way. This is a relaxation of a common assumption for many existing clustering algorithms that the input graph is simply well clusterable. We would like to very efficiently process such a noisy clusterable graph and extract useful information regarding its cluster structure. Slightly more precisely, we study two types of sublinear algorithms for analyzing the cluster structure of graphs with noisy partial information. 

The first type of algorithm is driven by the following natural question: \emph{Given a noisy clusterable graph, can we build an oracle (or implicit representation) in sublinear time, that can support typical queries regarding the cluster structure of the graph in sublinear time?} For example, we would like to query ``Is a vertex $s$ a noise/outlier?''. If the answer is ``No'', we would further like to know ``Which cluster does $s$ belong to?'', and ``Do $s$ and $t$ belong to the same cluster?'', given that both vertices $s,t$ are not outliers. We would require that all the query answers will be consistent, e.g., if $u,v$ are reported to belong to the same cluster, $v,w$ are reported to belong to the same cluster, then $u,w$ will also be reported to belong to the same cluster. Furthermore, we would like to minimize the number of vertices for which the oracle returns the ``wrong'' answers in the sense that the output partition of the algorithm should be close to an underlying maximal good clustering of the graph. We will call such an oracle as a \emph{robust clustering oracle}. Such oracles might be already interesting from real-world applications. For example, quickly identifying outliers might be valuable in road networks and medical data. Sometimes, we only want the cluster information of a small group of vertices while do not care about other parts of the graph. Furthermore, it will be desirable to work on-the-fly on a clean data after removing a small fraction of outliers. Besides these real-world applications, such oracles might be given as input for other clustering algorithms that are equipped with the power of making the above mentioned clustering queries (see e.g., \cite{mazumdar2017query,MS17:clustering,AKB16:clustering,ailon2018approximate,ABJ18:approximate}). %

Our second type of algorithm is motivated by a very related question: \emph{Given a noisy clusterable graph, can we fix it by minimally modifying the original graph, and provide query access to the reconstructed clusterable graph in sublinear time?} We address this question in the \emph{online reconstruction} framework introduced by~\cite{ACCL08:reconstruction}. In this framework (for graphs), given a property $\Pi$ and query access to a graph $G$ that is close to have $\Pi$, we want to output a graph $G'$ such that $G'$ has the property $\Pi$ and $G$ is modified minimally to get $G'$. Furthermore, we would like to output $G'$ in a local and consistent way that can provide query access to $G'$ by making as few queries to the input graph $G$. The corresponding algorithm will be called a \emph{local reconstructor} or \emph{local filter} for property $\Pi$~\cite{ACCL08:reconstruction,SS10:monotonicity,AT10:testability}. The natural application of such local reconstructors is when only a small portion of the corrected graph $G'$ is needed or when we want to make use of the graph $G'$ in a distributed manner. (Note that in many applications, queries are made to a large graph which are assumed to exhibit some structural property.) Here, we would focus on designing a local filter for cluster structure of graphs and providing consistent query access to a clusterable graph. In practice, such algorithms might be used for fast recommending products to users even if there are some noise in the data. %

In this work, we give both sublinear robust clustering oracle and local reconstructors for the cluster structure of graphs. Now we give basic definitions of clusters and (noisy) clusterable graphs, formalize our algorithmic problems, state our main results and sketch our technical ideas.

\subsection{Basic Definitions}
\paragraph{Conductance based clustering.}
Following a recent line of research on graph clustering~(e.g., \cite{OT14:partitioning,CPS15:cluster,PSZ17:partition,DPRS14:spectra}, which were built upon \cite{KVV04:clustering}), we will use \emph{conductance} based definition for measuring the quality of clusters and the cluster structure of graphs. In this paper, we will focus on undirected graphs with bounded maximum degree. We call an undirected graph $G=(V,E)$ a \emph{$d$-bounded} graph if its maximum degree is upper bounded by some parameter $d$, which is always assumed to be some sufficiently large constant (at least $10$). For any two subsets $S,T\subseteq V$, we let $E(S,T)$ denote the set of edges with one endpoint in $S$ and the other point in $T$. The \emph{conductance $\phi_G(S)$} of a set $S$ in $G$ is defined to be the ratio between the number of edges crossing $S$ and its complement $V\setminus S$ and the maximum number of edges possible incident to $S$, that is, $\phi_G(S):=\frac{|E(S,V\setminus S)|}{d|S|}.$ The \emph{conductance $\phi(G)$} of the graph $G$ is defined to be the minimum value of conductance of set $S$ with size at most $n/2$, that is, $\phi(G): = \min_{S:|S|\leq n/2} \phi_G(S).$ For convenience, for the singleton graph $G$ (that consists of a single vertex with no edges) we define its inner conductance $\phi(G)$ to be $1$.

Given a vertex set $S\subset V$, we let $G[S]$ denote the subgraph graph induced by vertices in $S$. In the following, we will refer to $\phi_G(S)$ and $\phi(G[S])$ as the \emph{outer conductance} and \emph{inner conductance}, respectively. %
Given two parameters $\phi_{\myin}$ and $\phi_{\myout}$, we call a set $S$ a \emph{$(\phi_{\myin}, \phi_{\myout})$-cluster} if 
$$\phi_G(S)\leq \phi_{\myout}, \quad \phi(G[S]) \geq \phi_{\myin}.$$
For a good cluster $S$, we expect $\phi_{\myin}$ to be large and $\phi_{\myout}$ to be small. In particular, if $S=V$ and $\phi(G[V])=\phi(G)\geq \phi_\myin\geq \phi$ for some constant $\phi$, then we call the graph $G$ a \emph{$\phi$-expander} which by itself is a good cluster and has been extensively studied in theoretical computer science~(see e.g.,~\cite{HLW06:expander}). It is useful to note that $\phi_G(V)=0$. When $G$ is clear from the context, we omit the subscript $G$ from $\phi_G(S)$. %
A \emph{$k$-partition} of a graph $G=(V,E)$ is a partition of $V$ into $k$ subsets, $V_1,\cdots, V_k$ such that $V_i\cap V_j=\emptyset$ for $i\ne j$ and $\cup_i V_i=V$. We have the following definition of clusterable graphs that characterize graphs with typical cluster structure (see e.g., \cite{OT14:partitioning}).

\begin{definition}
Given parameters $d, k,\phi_\myin,\phi_\myout$, we call a $k$-partition $P_1,\cdots,P_k$ of a $d$-bounded graph $G$ a \emph{$(k,\phi_\myin,\phi_\myout)$-clustering} if for each $i\leq k$, $\phi(G[P_i])\geq \phi_\myin$ and $\phi_G(P_i) \leq \phi_\myout$.

A $d$-bounded graph $G$ is called to be \emph{$(k,\phi_\myin,\phi_\myout)$-clusterable} if $G$ has an $(h,\phi_\myin,\phi_\myout)$-clustering for some $h\leq k$.
\end{definition}
Note that in our definition, a $(k,\phi_\myin,\phi_\myout)$-clusterable graph may contain less than $k$ clusters, and $(1,\phi_\myin,0)$-clusterable graphs are equivalent to $\phi_\myin$-expanders. 

\paragraph{Clusterable graphs with modeling noise.} We assume that the input graph to the algorithm is generated from the family of all $(k,\phi_\myin,\phi_\myout)$-clusterable graphs and then modified by an adversary in some manner. %
We have the following definition.
\begin{definition}{(Clusterable Graphs with Modeling Noise or Noisy Clusterable Graphs)}
In this model, the adversary first chooses an arbitrary graph $G^*$ from the family of all $(k,\phi_\myin,\phi_\myout)$-clusterable graphs with maximum degree upper bounded by $d$. Then the adversary may do the following:
\begin{enumerate}
\item Choose an arbitrary $(h,\phi_\myin,\phi_\myout)$-clustering $P_1,\cdots,P_h$ of $G^*$ for some $h\leq k$.
\item Insert and/or delete at most $\varepsilon \cdot d n$ edges (noise) within the clusters $G^*[P_i]$, $1\le i\leq h$, while preserving the degree bound.
\end{enumerate}
We call the resulting graph $G$ an \emph{$\varepsilon$-perturbation} of $G^*$ with respect to the $h$-partition $P_1,\cdots, P_h$. %

\end{definition}
Equivalently, a graph $G$ is called to be an $\varepsilon$-perturbation of a $(k,\phi_\myin,\phi_\myout)$-clusterable graph if there is partition of $G$ with at most $k$ parts (called clusters), such that one can insert/delete at most $\varepsilon dn$  \emph{intra-cluster} edges to make it a $(k,\phi_\myin,\phi_\myout)$-clusterable graph. For simplicity, in the above definition, we only allowed the adversary to perturb the edges \emph{inside} the clusters, while our algorithm can actually be extended to work for the case that the adversary is also allowed to perturb \emph{inter-cluster} edges, up to a very \emph{limited extent}\footnote{More precisely, the adversary can be allowed to perturb a $\phi_\myout$ fraction of inter-cluster edges: this essentially can then be reduced to the case that only intra-cluster perturbations are allowed by re-scaling a constant factor of conductance values, i.e., one can view that the adversary first chooses a $(k,\phi_{\myin},2\phi_\myout)$-clusterable graph and then perturbs its intra-cluster edges.}. %
This definition generalizes the notion of noisy expander graphs studied by Kale, Peres, and Seshadhri~\cite{KPS13:noise}, which correspond to $k=1$ in our problem. In their setting, the adversary first chooses a $\phi$-expander and then modifies it by inserting/deleting  $\varepsilon$ fraction of edges in the graph.

\subsection{Problem Formalizations and Main Results}
Now we formalize our algorithmic problems and present our main results. For a $d$-bounded graph $G$, we will assume the algorithm is given query access to the adjacency list of $G$, that is, in constant time we can query the $i$-th neighbor of any vertex $v$. 

\paragraph{Robust clustering oracle.} Given query access to the adjacency list of a $d$-bounded graph $G$ that is promised to be an $\varepsilon$-perturbation of a $(k,\phi_\myin,\phi_\myout)$-clusterable graph, we are interested in constructing an implicit representation, called a \emph{robust clustering oracle}, of $G$ in sublinear time such that typical queries regarding the cluster structure of $G$ can be answered as quickly as possible (also in sublinear time).  More precisely, the oracle should support the following types of \emph{clustering queries}: 
\begin{itemize}
\vspace{-0.2em}
\setlength{\topsep}{0pt}
\setlength{\itemsep}{0pt}
\item[1)] \textsc{IsOutlier}($s$): Is a vertex $s$ a noise/outlier?
\vspace{-0.2em}
\end{itemize}
Intuitively, a vertex that does not belong to any good cluster should be reported as noise or outlier. For any non-outlier vertices $s,t$, the oracle can further support
\begin{itemize}
\setlength{\topsep}{0pt}
\setlength{\itemsep}{0pt}
\item[2)] \textsc{WhichCluster}($s$): Which cluster does $s$ belong to?
\item[3)] \textsc{SameCluster}($s,t$): Do $s$ and $t$ belong to the same cluster? 
\end{itemize}  

In the following, without loss of generality, we will assume that for any non-outlier vertex $s$ and the corresponding \textsc{WhichCluster}($s$) query, the oracle will output an integer $i$ with $1\leq i\leq h$ that specifies the index of the cluster that $s$ belongs to, for some integer $h$. Furthermore, given the ability of answering \textsc{WhichCluster} queries, for any two non-outlier vertices $s,t$, we simply define \textsc{SameCluster}($s,t$) to be the procedure that checks if \textsc{WhichCluster}($s$) is equal to \textsc{WhichCluster}($t$). This will naturally ensures the consistency for \textsc{SameCluster} queries. Note that the output of the algorithm naturally defines a partition of $V$, i.e., 
\[P_i:=\{u\in V: \textsc{WhichCluster}(u)=i\}, 1\leq i\leq h, \quad  %
B:=\{u\in V: \textsc{IsOutlier}(u)=\textbf{Yes}\}.\]

We would like to minimize the number of vertices for which the oracle returns the ``wrong'' answers. That is, for most vertices $v$ that do belong to some underlying good cluster in the perturbed $G$, we expect \textsc{IsOutlier}($v$) to return ``No''. Furthermore, for most vertices $u,v$ that belong to the same cluster (resp. different clusters), we expect \textsc{SameCluster}($u,v$) to return ``Yes'' (resp. ``No''). One further crucial requirement of a robust clustering oracle and the corresponding clustering query algorithm is to maintain \emph{consistency} among all queries. That is, on different query sequences, the answers of the oracle should be consistent with the same $h$-partition $D_1,\cdots, D_h$ of $V$ for some $h\leq k$, in which all but a small fraction of vertices belong to some \emph{good} cluster. Since the oracle construction and the corresponding query algorithm are typically randomized, we fix the randomness seed of the oracle and query algorithm once and for all to ensure consistent answers. Then the algorithm will be a \emph{deterministic} procedure for any input query, which further guarantees that the partition $D_1,\cdots, D_h$ is determined by $G$ and the internal randomness of the oracle and the algorithm, and is independent of the order of queries. This feature allows the oracle to be used in the distributed manner as consistency is guaranteed.

We provide the first robust clustering oracle with both sublinear preprocessing time and query time. For simplicity, we will assume both $d,k$ are constant throughout the paper. %
Let $P\triangle Q$ denote the symmetric difference between two vertex sets $P,Q$.

\begin{theorem}[Robust Clustering Oracle]\label{thm:oracle}
There exists an algorithm that takes as input parameters $n\geq 1$, $d>10$, $k\geq 1$, $\phi\in (0,1)$, $\varepsilon\in [\Omega(\frac{\phi}{{n}}),1]$ and has query access to the adjacency list of a graph $G=(V,E)$ that is an $\varepsilon$-perturbation of a $(k,\phi,O(\frac{\varepsilon\phi}{k^3\log n}))$-clusterable graph, and constructs a robust clustering oracle in $O(\sqrt{n}\cdot \poly(\frac{k\cdot\log n}{\phi\varepsilon}))$ pre-processing time. Furthermore, it holds that
\begin{enumerate}
\item Using the oracle, the algorithm can answer any clustering query (i.e., \textsc{IsOutlier}, \textsc{WhichCluster} or \textsc{SameCluster}) in $O(\sqrt{n}\cdot \poly(\frac{k\cdot\log n}{\phi\varepsilon}))$ time. 
\item There exists a partition $D_1,\cdots, D_{h'}, B'$ of $G$, for some $h'\leq k$, such that 
\begin{itemize}
	\setlength{\itemsep}{0pt}
\item the partition only depends on $G$ and the input parameters of the algorithm, and is independent of the order of queries;  
 \item if $\varepsilon\in [\Omega(\frac{\phi}{{n}}),\frac{\phi}{60k^2}]$, then
	$h'\geq 1$ and each $D_i$ is a $(\frac{\phi}{2}, \frac{a_{\ref{thm:rw_perturbed}}\sqrt{\varepsilon}\kappa^4\phi^{1.5}}{3k^3\log n})$-cluster, for any $1\leq i\leq h'$; if $\varepsilon\in (\frac{\phi}{60k^2},1]$, then $h'=0$; and
	\item with probability at least $1-\frac{1}{n}$, the partition $P_1,\cdots,P_{h},B$ output by the algorithm satisfies that $h'\leq h\leq k$ and  $\sum_{i=1}^{h'}|P_i\triangle D_i| + |(\cup_{i=h'+1}^{h}P_i)\cup B|+|B'|= O(k\sqrt{\frac{\varepsilon}{\phi}}n)$. 
\end{itemize}

\end{enumerate}

\end{theorem}

We remark that there is no algorithm that allows both $o(\sqrt{n})$ pre-processing time and $o(\sqrt{n})$ query time for \textsc{IsOutlier} queries, as otherwise, one could obtain a property testing algorithm for expansion with $o(\sqrt{n})$ queries, which will be a contradiction to a known lower bound~\cite{GR00:expansion} (see more discussions below on relation to property testing). Furthermore, the second item of the theorem implies that the total number of vertices that are reported as outliers is at most $O(k\sqrt{\frac{\varepsilon}{\phi}}n)$ %
and that the query answers are consistent with a partition of $G$ in which all but $O_k(\sqrt{\frac{\varepsilon}{\phi}}n)$ vertices belong to a $(\frac{\phi}{2},O_k(\frac{\sqrt{\varepsilon}\phi^{1.5}}{\log n}))$-cluster. We also note that in the statement of the above theorem, the most interesting range of $\varepsilon$ is\footnote{Note that in this range, $\varepsilon=O(\frac{\phi}{k^2})$, which is also the reason that we do see the traditional $\phi^2$ dependency (from Cheeger's inequality) between the outer conductance and inner conductance.}  $\varepsilon\in[\Omega(\frac{\phi}{{n}}),\frac{\phi}{60k^2}]$, as otherwise (i.e., $\varepsilon>\frac{\phi}{60k^2}$) the noise will be too much and our algorithm cannot guarantee to locally identify even one cluster. Removing the $\log n$ gap between the inner conductance and outer conductance seems to be hard, at least for methods that are based on random walk distances (as we used here). For example, %
in \cite{CKKMP18:cluster}, it has been discussed that in general, it is impossible to use Euclidean distance between random walk distributions to test $2$-clusterablity if one wants the gap to be a constant. (Testing $2$-clusterability is an easier problem than the robust clustering oracle problem; see below.) On the other hand, being able to correctly answer \textsc{SameCluster}($u,v$) queries intuitively requires or induces a distance based approach, as the vertices in the same cluster are ``similar'' or ``close to'' each other, while vertices in different clusters are ``dissimilar'' or ``far from'' each other.

%


%

%

\paragraph{Local reconstructor of graph cluster structure.} %
We are interested in designing a local reconstruction algorithm for the cluster structure of graphs. Given query access to the adjacency list of a $d$-bounded graph $G$ that is promised to be an $\varepsilon$-perturbation of a $(k,\phi_\myin,\phi_\myout)$-clusterable graph, our goal is to design a \emph{local filter} that provides query access to a $(k,\phi_\myin',\phi_\myout')$-clusterable graph $G'$ such that the distance between $G$ and $G'$ is as close as possible. That is, we would like to output $G'$ in a local manner that for any vertex query, the neighborhood of $v$, i.e., the set of all neighbors of $v$, in $G'$ can be answered in sublinear time (in particular, by making as few queries to the adjacency list to $G$ as possible). Similar as for the robust clustering oracle, it is crucial to require a local filter to maintain \emph{consistency} among all queries. Here we require that for different query sequences, the answers of the filter should be consistent with the same reconstructed graph $G'$. %
Again, the filter is suitable to be used in the distributed manner as consistency is guaranteed. In our local filter for clusterable graphs, we also aim to make the gap between $\phi_\myin,\phi_\myout$ and the gap between $\phi_\myin$ and $\phi_\myin'$ as small as possible. We next state our theorem regarding our local filter for clusterable graphs as follows. %
\begin{theorem}[Local Reconstructor of Cluster Structure]
\label{thm:main}
There exists a local reconstruction algorithm that takes as input parameters $n\geq 1$, $d>10$, $k\geq 1$, $\phi\in (0,1)$, $\varepsilon\in [\Omega(\frac{\phi}{{n}}),  1]$ and has query access to the adjacency list of a graph $G=(V,E)$ that is an $\varepsilon$-perturbation of a $(k,\phi,O(\frac{\varepsilon\phi}{k^3\log n}))$-clusterable graph, and provides query access to a graph $G'=(V,E')$ such that the following holds with probability at least $1-\frac{4}{n}$:
\begin{enumerate}
\item\label{item:innerouter} $G'$ is $(k,\Omega(\frac{\varepsilon\phi}{k^4\log n}),1)$-clusterable, and has maximum degree at most $d+16$.
\item The number of edges changed is at most $O(\min\{1,k\sqrt{\frac{\varepsilon}{\phi}}\}\cdot n)$. %
\item\label{item:consistent} $G'$ is determined by $G$ and the internal randomness of the algorithm, and is independent of the order of queries.
\item On each query $v$, the neighborhood of $v$ in $G'$ can be answered in $O(\sqrt{n}\cdot \poly(\frac{k\cdot\log n}{\phi\varepsilon}))$ time.
\end{enumerate}
\end{theorem}

Note that by Item~\ref{item:innerouter}, the resulting graph can be partitioned into at most $k$ parts, each with relatively large inner conductance (i.e.,  $\Omega_k(\frac{\varepsilon\phi}{\log n})$), with no guarantee on outer conductance (as each set trivially has outer conductance at most $1$). (Such instances are exactly the object that was studied in~\cite{CKKMP18:cluster} in the framework of property testing.) By sacrificing the inner conductance quality, we can also find a clustering of $G'$ with small outer conductance. That is, we can guarantee that $G'$ is also $(k,\Omega(\frac{\nu^k}{6^kk^4}\frac{\varepsilon\phi}{\log n}),\min\{k\nu, 1\})$-clusterable for any $\nu\in[0,1]$ (see Appendix~\ref{sec:cluster_outersmall} for details). %
Item~\ref{item:consistent} implies that all query answers are consistent, that is, the vertex $u$ is output as a neighbor of $v$ in $G'$ if and only if $v$ is output as a neighbor of $u$. From the discussion below on the connections between our local reconstruction algorithm and property testing, the running time of our filter is optimal (in terms of dependency on $n$) up to polylogarithmic factors. 
	
Furthermore, our algorithm generalizes the local reconstruction algorithm for expander graphs by \cite{KPS13:noise}, which corresponds to the special case $k=1$ in our problem, though our approximation ratio of the number of modified edges is worse. More precisely, for $\varepsilon=\Omega(\phi)$, both our algorithm and the algorithm in~\cite{KPS13:noise} will add $\Theta(dn)$ edges (as the noise part is too large, and thus almost all vertices will be reported as outliers and the resulting graph is almost the complete hybrid of the original graph and an explicitly constructible expander (see Section~\ref{sec:techniques} for more discussions)); for $\varepsilon=O(\phi)$, the algorithm in~\cite{KPS13:noise} reconstructs a graph that is an $\varepsilon$-perturbation of a $\phi$-expander by modifying at most $O(\frac{\varepsilon}{\phi}n)$ edges, and the resulting graph has conductance at least $\Omega(\frac{\phi^2}{\log n})$ and maximum degree also upper bounded by\footnote{Note that \cite{KPS13:noise} claimed that the number of modified edges is at most $O(\frac{\phi}{\log n}\varepsilon n)$ and the maximum degree of the resulting graph is $d+O(\lceil\frac{d\phi^2}{\log n}\rceil)$. However, this claim is not correct (at least for $d$-bounded graphs with $d$ being constant), and the number of changed edges and the maximum degree bound from their analysis should be $O(\frac{\varepsilon}{\phi}n)$ and $d+16$, respectively~\cite{Sesh}. They obtained their claimed results by adding $t:=\lceil \frac{d\phi^2}{c\log n}\rceil$ parallel edges while repairing bad vertices, from which they get that the maximum degree is $d+16t$ and the number of added edges to the optimal distance (i.e., $\varepsilon d n$) is $\frac{16t}{d\phi}=O(\phi/\log n)$, which is incorrect as it always holds that $t=1$ for constant $d$ and large enough $n$.} $d+16$, while our algorithm has to modify $O(\sqrt{\frac{\varepsilon}{\phi}}\cdot kn)$ edges. We further note that the algorithm in~\cite{KPS13:noise} guarantees that the reconstructed graph has inner conductance at least $\Omega(\frac{\phi^2}{\log n})$, while the resulting graph from our algorithm is guaranteed to have a partition with at most $k$ parts, each with inner conductance at least $\Omega_k(\frac{\varepsilon\phi}{\log n})$. Removing the $\log n$ factor in the inner conductance of the output graph seems to be a very challenging task, even for the case $k=1$. See Section~\ref{sec:conclusions} for more discussions. %

\paragraph{Local mixing property on noisy clusterable graphs.} In order to derive the above algorithmic results, we prove an interesting behavior, which we call \emph{local mixing property}, of random walks on noisy clusterable graphs. For technical reasons, we will consider the \emph{uniform averaging walk of $t$ steps} on a graph $G$: In this walk, we choose a number $\ell\in\{0,1,2,\cdots,t-1\}$ uniformly at random, and stop the (normal) random walk after $\ell$ steps. We let $\pp_v^t$ denote the probability vector for a uniform averaging walk of $t$ steps starting at $v$ and let $\norm{\p_1-\p_2}_{\TV}$ denote the total variance distance between two distributions $\p_1,\p_2$. We have the following theorem.
\begin{theorem}[Local Mixing Property of Random Walks]\label{thm:rw_perturbed}
	Let $0<\ccss,\varepsilon<1$. Let $\phi_\myout\leq \frac{a_{\ref{thm:rw_perturbed}}\varepsilon\ccss^4\phi_\myin^2}{k^3\log n}$ for some sufficiently small constant $a_{\ref{thm:rw_perturbed}}>0$. Let $G$ be a $d$-bounded graph with an $h$-partition $C_1,C_2,\cdots, C_h$ such that $\phi_G(C_i)\leq \phi_\myout$ for any $1\leq i\leq h\leq k$. %
	For each $i\leq h$, we let $D_i\subseteq C_i$ denote a large subset of vertices such that $\phi(G[D_i])\geq \phi_\myin$, and let $B_i：=C_i\setminus D_i$. If $\sum_i|B_i|\leq \varepsilon n$, then for any $D_j$ with $|D_j|\geq 3\sqrt{\varepsilon} n$, there exists a subset $\widehat{D}_j\subseteq D_j$ such that $|\widehat{D}_j|\geq (1-4\sqrt{\varepsilon})|D_j|$ such that for any $s\in \widehat{D}_j$, and $t= \frac{120\log n}{\ccss\phi_\myin^2}$, it holds that 
	$$\norm{\pp_{s}^{t}- \uni_{C_j}}_{\TV}<\ccss +\sqrt{\varepsilon}.$$
\end{theorem}

Intuitively, the set $B_i$ corresponds to the noisy part inside each cluster $C_i$ and we assume that the total fraction of noisy part is parametrized by $\varepsilon$. Then the above theorem says that the rest of the large part (i.e., clusterable part) exhibits some nice local mixing property: a typical uniform averaging random walk (of appropriately chosen length) from a large cluster (of size $\Omega(\sqrt{\varepsilon}n)$) will converge quickly to the uniform distribution on it. %
This is a generalization of the global mixing property of noisy expander graphs in \cite{KPS13:noise}, though their results are stated for the more general Markov chains. 

\subsection{Our Techniques}\label{sec:techniques}
To design a robust clustering oracle, we first note that it is relatively easy 
to design a clustering oracle without noise (if the gap between $\phi_\myin$ and $\phi_\myout$ is $O(\log n)$ as we considered here). This can be done by a refined analysis of the property testing algorithm in \cite{CPS15:cluster} that samples a small number of vertices, and then test if the $\ell_2$ norm distance between the random walk distributions from any two vertices is larger than some threshold or not. However, the analysis depends on the spectral property (e.g., a gap between $\lambda_k$ and $\lambda_{k+1}$) of clusterable graphs, and cannot be easily generalized to the case that the input graph contains noise, as such spectral property is very sensitive to noise (e.g., deleting all edges incident to a constant number of vertices will break down the property). 

In order to handle noisy input, we use the $\ell_1$ norm distance between the corresponding random walk distributions to test if the starting two vertices belong to the same cluster or not, and we make use of the local mixing property of random walks in Theorem~\ref{thm:rw_perturbed}.  %
In order to prove the such a mixing property, we first show that it does hold for clusterable graphs \emph{without} noise, by exploiting a spectral property that characterizes the first $k$ eigenvectors of clusterable graphs given by \cite{PSZ17:partition}.   %
To generalize the result to a noisy clusterable graph $G$, we 
view the random walks on the graph as a Markov chain and consider a new Markov chain that is induced on vertices in the clusterable part in $G$. (Such a new chain has also been used in~\cite{KPS13:noise} for analyzing noisy expanders.) We show the induced Markov chain does correspond to a clusterable graph $H$ (by overcoming the difficulty that the outer conductance of each corresponding cluster increases and might change the cluster structure too much) and thus the random walks in $H$ satisfy the local mixing property. However, the walks on $H$ can be very different from the random walks in the original graph $G$. 
We then give a novel application of an old technique called \emph{stopping rules} of Markov chains that was introduced by Lov{\'o}sz and Winkler \cite{LW97:mixing} to relate these two walks, and bound the total variance distance between two random walk distributions from a vertex in any large cluster of $G$ and $H$. %
This allows us to show the local mixing property in the graph $G$. To the best of our knowledge, we are the first to use of the tool of stopping rules to show that a random walk in the graph mixes inside a \emph{subgraph} (i.e., cluster) rather than in the whole graph. 

Given such a local mixing property of random walks in the noisy clusterable graph, we are able to design a robust clustering oracle and the corresponding clustering query algorithm with sublinear preprocessing and query time. We first note that if the noisy part is not too large (i.e., $\varepsilon=O(\phi/k^2)$), then the graph $G$ has a non-trivial partition $D_1,\cdots,D_{h'},B'$ with $h'\geq 1$ that only depends on the corresponding parameters (i.e., $\varepsilon, \phi, n$) and $G$ itself, and that each $D_i$ is a good cluster with large size (containing at least $\Omega(\sqrt{\varepsilon})$ fraction of vertices), and $B'$ has small size. Our key idea is to use random walks to learn a succinct representation $H$, which is a weighted graph with roughly $O(\log n)$ vertices, of the clusterable part of graph $G$, such that each cluster $D_i$ in $G$ will be mapped to a unique clique (called a \emph{core}) in $H$ with appropriate edge weight. Furthermore, by using the weights and the size bounds of these cliques, we can be efficiently identify them from $H$, using which we are able to answer the \textsc{WhichCluster} queries. Slightly more precisely, in the \emph{preprocessing} (or \emph{learning}) phase, the algorithm samples a set $S$ of $\Theta(\log n)$ vertices, and uses the statistics of $\tilde{O}(\sqrt{n})$ random walks from each sampled vertex to (quite accurately) estimate the so-called \emph{reduced collision probability (rcp)} of (the random walks of appropriate length from) any two sampled vertices that was introduced in~\cite{KPS13:noise}. We construct a weighted \emph{similarity graph} $H$ on the sample set $S$ such that the weight of each edge $(u,v)$ is our estimate of the rcp of $u,v$, for any $u,v\in S$. We show that if the noisy part is not too large, then, by the aforementioned local mixing property, for (most) pair of vertices $u,v\in S \cap D_i$, the rcp of $u,v$ will be close to $1/|D_i|$. Thus, the weight of edge $(u,v)$ in $H$ will be set to be a number close to $1/|D_i|$, and most vertices in $S\cap D_i$ form a clique $S_i$ in $H$ with edge weights close to $1/|D_i|$. We further observe that $S_i$ has relatively large size (roughly $|S|\cdot \frac{|D_i|}{n}$), as $|D_i|$ is large; and that any vertex $v\in S_i$ can only belong to exactly one such (large) clique, as otherwise, the total probability mass of random walk distribution from $v$ will exceed $1$, which can not happen. These properties allow us to efficiently identify the unique core $S_i$ from $H$ that corresponds to the cluster $D_i$  by a simple greedy algorithm and further to answer membership queries.  
We remark that in~\cite{CPS15:cluster}, a similarity graph is also constructed, while that graph is unweighted and only tells if the original graph is $k$-clusterable or not according to the number of connected components, which is far from  sufficient for our application.




Then in the \emph{query} phase, we check if the queried vertex $v$ belongs to any of the learned cores or not to decide if it is an outlier or not. This, again, can be done by estimating the rcp of the walks from $v$ and other vertices in $S$ (by running $\tilde{O}(\sqrt{n})$ random walks), and is guaranteed by the local mixing property of random walks. In particular, for most vertices $v$ in a cluster $D_i$, the rcp of random walks from $v$ and any other vertex that is in $S_i$ corresponding to $D_i$ will be also around $1/|D_i|$. If this is the case, we output $i$ as the index of the cluster that $v$ belongs to; otherwise, we report it as an outlier. %
The above analysis shows that most vertices in $D_1,\cdots, D_{h'}$ will be correctly classified, or equivalently, the number 
of vertices that are reported as outliers is small. 

Our local reconstruction algorithm for clusterable graphs is built upon our robust clustering oracle. That is, we first learn the cores of the input graph as before. Then (if the noisy part is not too large) we only ``repair'' all the vertices that are reported as outliers. Let $v$ be any vertex that is reported as an outlier. We add all the neighbors of $v$ in an explicit expander $G_{\exp}$ to ``repair'' the graph $G$, which is called a \emph{hybridization} (between $G_{\exp}$ and $G$) and has been used to repair expander graphs in \cite{KPS13:noise}. Then the answers is guaranteed to be consistent with a graph $G'$ such that its distance to the original graph $G$ is at most $d$ times the number of vertices that are reported as outliers, which has already been bounded to be small. In order to prove the claimed guarantee on cluster structure of $G'$, we introduce a definition of \emph{weak} vertices that intuitively correspond to the noisy part of the graph. Such a definition has also been used in~\cite{KPS13:noise}, though ours is more subtle, depending on the size of noise. We can show that one can improve the cluster structure of the graph if we have repaired all the weak vertices in the above way. Furthermore, such weak vertices will always be reported as outliers, which is guaranteed by the performance of our robust clustering oracle. %

\subsection{Relation to Testing Graph Clusterability} %
Both the above robust clustering oracle and local reconstruction are closely related to the framework of \emph{property testing}~\cite{RS96:testing,GGR98:property}. In the bounded degree graph property testing~\cite{GR02:property}, given a property $\Pi$,  the algorithm aims to distinguish graphs that satisfy $\Pi$ from graphs that are $\varepsilon$-far from satisfying $\Pi$ by making as few queries (to the adjacency list of the graph) as possible, with high constant probability, say at least $2/3$. Here, a graph is said to be $\varepsilon$-far from satisfying property $\Pi$ if one has to modify more than $\varepsilon dn$ edges to make it satisfy $\Pi$, while preserving the degree bound. After two decades of study, a number of properties of bounded degree graphs are now known to be testable in constant time~\cite{GR02:property,BSS10:every,HKNO09:local,NS13:every}, $\tilde{O}(\sqrt{n})$ or $\tilde{O}(n^{\frac{1}{2}+c})$ time~\cite{GR98:sublinear,GR00:expansion,CS10:expansion,KS11:expansion,NS10:expansion,CPS15:cluster,CKKMP18:cluster,KSS18:minor}. 

In particular, for the property of being $(k,\phi_\myin,\phi_\myout)$-clusterable, \cite{CPS15:cluster} gave a testing algorithm that runs in time $\tilde{O}(\sqrt{n}\poly(\phi,k,1/\varepsilon))$ and distinguishes $(k,\phi,O(\frac{\phi^2\varepsilon^4}{k^{\Omega(1)}}))$-clusterable graphs from graphs that are $\varepsilon$-far from being $(k,\Theta(\frac{\phi^2\varepsilon^4}{k^{\Omega(1)}\log n}), \psi)$-clusterable, for any $\psi\in [0,1]$. (Note that the algorithm rejects any graph that is far from clusterable graphs with \emph{arbitrary} outer conductance.) \cite{CKKMP18:cluster} recently improved this algorithm by giving an algorithm for testing if a graph contains at most $k$ subsets with inner conductance at least $\phi$ from those that can be decomposed into at least $k+1$ subsets with size at least $\Omega_k(n)$ and outer conductance at most $O(\mu\phi^2)$ in time $O(n^{1/2+O(\mu)})$ for any $\mu$ that is smaller than some constant (they also generalize their algorithm for general graphs). For the case of $k=1$, i.e., testing if the graph has expansion at least $\phi$, the best known algorithm can test if a graph has expansion $\phi$ or is $\varepsilon$-far from having expansion $\Theta(\mu\phi^2)$ in time $\tilde{O}(n^{0.5+\mu})$ for any $\mu>0$~(\cite{KS11:expansion,NS10:expansion} which improves upon \cite{CS10:expansion}). Furthermore, there exists a lower bound of $\Omega(\sqrt{n})$ on the query complexity for testing expansion~\cite{GR02:property}. %

Note that both the robust clustering oracle problem and the reconstruction problem are always much harder than the property testing version (see e.g., \cite{KPS13:noise}). For example, in the oracle problem, we need to figure out the cluster structure of the clusterable graph, and in the local reconstruction problem, the algorithm actively repairs the input graph, while the property testing is a decision problem. Furthermore, property testing only needs to distinguish between graphs which are clusterable and those are $\varepsilon$-far from being clusterable, while both the clustering oracle and the reconstruction have to (in some sense) approximate the distance to the class of all clusterable graphs\footnote{Actually, in our setting, we are approximating the \emph{intra-perturbation distance} to the class of all clusterable graphs, i.e., the minimum number of \emph{intra-cluster} edges needed to be modified to obtain a clusterable graph over all possible $h$-partitions, for some $h\leq k$. This is in contrast to approximating the distance to all clusterable graphs, which is the minimum number of edges needed to be modified to obtain a clusterable graph.}. Thus, the property testing algorithms can not be directly used to or easily modified to give a robust clustering oracle or local reconstruction algorithm. In particular, even for the case that the input graph is clusterable, one cannot use the corresponding property testing algorithm (on the clusterable graph) to answer \textsc{SameCluster} queries. Actually, both algorithms in~\cite{CPS15:cluster,CKKMP18:cluster} make decisions based on some \emph{small summarizations} of the input graph which are constructed by a small sample of vertices and the corresponding random walk statistics. Such small summarizations can be used to distinguish if the graph is $k$-clusterable or is far from being $k$-clusterable. However, if the graph is indeed $k$-clusterable, they cannot be used to distinguish if two vertices are from the same cluster or are from two different clusters. As we mentioned before, in \cite{CKKMP18:cluster}, evidence has been provided that in general it is not possible to use pairwise Euclidean distances between two random walk distributions to distinguish between $2$-clusterable graphs and far from $2$-clusterable graphs if the gap between conductances is constant.  %

On the other hand, property testing algorithms can always be obtained from the corresponding local reconstruction ones (which has already been noted in previous work on local reconstruction) and testing $k$-clusterability can also be obtained from our robust clustering oracle algorithm. This is also true in our scenario since we can estimate the distance between $G$ and a clusterable graph $G'$ with small additive error by sampling a constant number of vertices and running the oracle and clustering query algorithm (or the local reconstruction algorithm) on each sampled vertex to obtain the fraction of outlier vertices. We further note that if a graph $G$ is $\varepsilon$-far from any $(k,\phi_\myin,\phi_\myout)$-clusterable graph, then it cannot be an $\varepsilon$-perturbation of any such clusterable graph (i.e., one has to perturb more than an $\varepsilon$-fraction of edges). Therefore, both our robust clustering oracle and local reconstructor algorithm lead to a property testing algorithm that distinguishes $(k,\phi,O(\frac{\varepsilon\phi}{k^3\log n}))$-clusterable graphs from graphs that are $\varepsilon$-far from being $(k,\Omega(\frac{\nu^k}{6^kk^4}\frac{\varepsilon\phi}{\log n}),k\nu)$-clusterable for any $\nu\in [0,1]$, with probability at least $2/3$. The running time of the algorithm is $\tilde{O}(\sqrt{n})$, which is optimal up to polylogarithmic factors due to the $\sqrt{n}$ lower bound on the number of queries for testing expansion (corresponding to $k=1$ in our problem)~\cite{GR02:property}.

\subsection{Other Related Work}
The study on \emph{local graph clustering}~\cite{ST13:local,ACL06:local,AP09:sparse,OT12:local,AOPT16:local,ZLM2013:local,OZ14:flow} is also closely related to our work. In this framework, the goal is to find a cluster from a specified vertex with running time that is bounded in terms of the size of the output set (and with a weak dependence on $n$). %
In the scenario where both inner and outer conductance are used for measuring the quality of clusters, \cite{ZLM2013:local} gave a local clustering algorithm that outputs a set with conductance at most $\tilde{O}(\min\{\sqrt{\phi_G(A)}, \phi_G(A)/\sqrt{\Conn(A)}\})$ where $A$ is the target set, and $\Conn(A)$ is the reciprocal (e.g., $\phi(G[A])^2/(\log \vol(A))$) of the mixing time of the random walk over the induced subgraph $G[A]$ on $A$ and $\vol(A)$ is the total degree of vertices in $A$. It is also shown that the conductance guarantee ${\phi_G(A)}/{\sqrt{\Conn(A)}}$ is \emph{tight} among (some class of) random-walk based local algorithms \cite{ZLM2013:local}. It might be interesting to note the logarithmic factor (i.e., $\log(\vol(A))$) dependency appeared in these guarantees. The performance guarantee has later been improved by \cite{OZ14:flow} using a flow-based local improvement algorithm that finds a set with conductance $\psi=O(\phi_G(A))$, volume $O(\vol(A))$ and runs in time $\tilde{O}(\vol(A)/\psi)$, where $A$ is the target set with $\Conn(A)/\phi_G(A)=\Omega(1)$. Note that the running times of these algorithms are sublinear only if the size (or volume) of the target set is small (say, at most $o(n)$), while in our setting, the clusters of interest have at least linear size (for any constant $\varepsilon$).

Fully or partially recovering the clusters in the noisy model has been extensively studied in the ``global algorithm regimes''. Examples include recovering the planted partition in \emph{stochastic block model} with modeling errors or noise~(e.g., \cite{CL15:robust,GV16:community,MPW16:robust,MMV16:learning}), \emph{correlation clustering} on different ground-truth graphs in the \emph{semi-random} model~(e.g.,~\cite{MS10:correlation,CJSX14:clustering,GRSY14:tight,MMV15:correlation}) %
and partitioning the graph in the \emph{average-case} model ~\cite{MMV12:approximation,MMV14:constant,MMV15:correlation}. All these algorithms run in at least linear time. %

Local reconstruction of some other properties have been investigated before. Such properties include expanders~\cite{KPS13:noise}, graph connectivity and diameter~\cite{CGR13:local}, bipartite and $\rho$-clique dense graphs~\cite{Bra08:restore}, geometric properties~\cite{CS11:geometric}, monotone functions~\cite{ACCL08:reconstruction,SS10:monotonicity}, Lipschitz functions~\cite{JR13:lipschitz} and low rank matrices and subspaces~\cite{DGK17:local}. This algorithmic framework is also closely related to local decodable codes (e.g.,~\cite{STV99:pseudorandom}) and local decompression~\cite{DLRR13:compression}. The local reconstruction model has been generalized to \emph{local computation} model by Rubinfeld et al.~\cite{RTVX11:local,ARVX12:space}, and a number of problems like maximal independent set, hypergraph coloring and maximum matching have been investigated in this model~\cite{RTVX11:local,ARVX12:space,MRVX12:converting,MV13:matching}.

\paragraph{Organization of the paper.} We give preliminaries in Section~\ref{sec:preliminaries}. In Section~\ref{sec:oracle}, we give the algorithm and the analysis for our robust clustering oracle and prove Theorem~\ref{thm:oracle}. Then, we give our local reconstruction algorithm, its analysis and prove Theorem~\ref{thm:main} in Section~\ref{sec:algorithm}. Both proofs for Theorem \ref{thm:oracle} and \ref{thm:main} will rely on the local mixing property of random walks in noisy clusterable graphs, i.e., Theorem~\ref{thm:rw_perturbed}, which we prove in Section~\ref{sec:proof_localmixing}. We conclude in Section~\ref{sec:conclusions}.

\section{Preliminaries}\label{sec:preliminaries}
Let $G=(V,E)$ denote an $n$-vertex undirected graph $G$ with maximum degree bounded by some constant $d$, where $V=[n]:=\{1,\cdots,n\}$. For each vertex $v$, we let $d_v$ denote its degree. Throughout the paper, all the vectors will be row vectors unless otherwise specified or transposed to column vectors. For a vector $\vect{x}$, we let $\norm{\vect{x}}_1:=\sum_{i}\abs{\vect{x}(i)}$ and $\norm{\vect{x}}_2:=\sqrt{\sum_i\vect{x}(i)^2}$ to denote its $\ell_1$ norm and $\ell_2$ norm, respectively. Let $\1_S$ denote the indicator vector of set $S$, that is $\1_S(u)=1$ if $u\in S$ and $0$ otherwise. Let $\1_v:=\1_{\{v\}}$. Let $\uni_S:=\frac{\1_S}{|S|}$ denote the uniform distribution on set $S$. For any set $X$ of vectors $\vect{x}_1,\cdots,\vect{x}_s$, we let $\Span(X)=\Span(\vect{x}_1,\cdots,\vect{x}_s)$ denote the linear span of $X$, that is $\Span(X)=\{\sum_{i=1}^s\mu_i\vect{x}_i| \mu_i\in \R \}$. For a vector $\vect{x}$ and a set $S$, we let $\vect{x}(S):=\sum_{v\in S}\vect{x}(S)$. For two distributions $\p_1$ and $\p_2$, we let $\norm{\p_1-\p_2}_{\TV}$ denote the total variance distance between $\p_1,\p_2$. It is known that $\norm{\p_1-\p_2}_{\TV}=\frac{1}{2}\norm{\p_1-\p_2}_{1}$.

\paragraph{Different types of random walks on $G$.} We will consider the following random walks.

(1) \emph{(Normal) random walk of $t$ steps}. In a (normal) random walk, at each step, suppose we are at vertex $v$, then we jump to a random neighbor with probability $\frac{1}{2d}$ and stay at $v$ with the remaining probability $1-\frac{d_v}{2d}$. We stop the walk after $t$ steps. We let $\p_v^t$ denote the probability vector for a $t$ step random walk starting at $v$. 

(2) \emph{Uniform averaging walk of $t$ steps}. In this walk, we choose a number $\ell\in\{0,1,2,\cdots,t-1\}$ uniformly at random, and stop the (normal) random walk after $\ell$ steps. We let $\pp_v^t$ denote the probability vector for a uniform averaging walk of $t$ steps starting at $v$.

(3) \emph{Uniform averaging walk of $t$ steps with two phases}. In this walk, we choose two integers $\ell_1,\ell_2\in \{0,1, 2,\cdots,t-1 \}$ uniformly at random, and stop the walk after $\ell_1+\ell_2$ steps. We let $\qq_v^t$ denote the probability vector for a uniform averaging walk of $t$ steps with two phases starting at $v$.

It is useful to note that for any two vertices $u,v$, 
$\qq_u^t(v)=\sum_{w\in V}\pp_u^t(w)\cdot \pp_w^t(v).$

\paragraph{A simple reduction: from $d$-bounded graphs to $d$-regular graphs.} Given a graph $G$ with maximum degree upper bounded by $d$, it will be very convenient to consider the $d$-regular graph $G'$ that is obtained by adding an appropriate number of self-loops (each with half weight) to each vertex so that every vertex has degree exactly $d$. Note that the (normal) random walk on $G$ we defined above is exactly the \emph{lazy} random walk of the graph $G'$. Let $\mat{A}$ denote the adjacency matrix of $G'$, and let $\mat{L}:=\mat{I}-\frac{1}{d}\mat{A}$ denote the normalized Laplacian matrix of $G'$. We let $0=\lambda_1\leq \lambda_2\leq \cdots\leq \lambda_n\leq 2$ denote the eigenvalues of $\mat{L}$ and let $\vect{v}_1, \vect{v}_2,\cdots, \vect{v}_n$ denote the corresponding orthonormal (row) eigenvectors. That is, $\mat{L}=\sum_{i}\lambda_i\cdot\vect{v}^T\cdot \vect{v}$. Note that the lazy random walk matrix corresponding to $G'$ is $\mat{P}:=\frac{\mat{I}+\frac{1}{d}\mat{A}}{2}=\mat{I}-\frac{\mat{L}}{2}$. This implies that the eigenvalues of $\mat{P}$ are $1=1-\frac{\lambda_1}{2}, 1-\frac{\lambda_2}{2}, \cdots, 1-\frac{\lambda_n}{2}\geq 0$, with corresponding eigenvectors $\vect{v}_1, \vect{v}_2,\cdots, \vect{v}_n$. In particular, $\mat{P}=\sum_{i}(1-\frac{\lambda_i}{2})\cdot \vect{v}^T\cdot \vect{v}$. Furthermore, it holds that $\p_v^t=\1_v\cdot \mat{P}^t=\sum_{i}(1-\frac{\lambda_i}{2})^t\cdot\vect{v}^T\cdot\vect{v}$.

\paragraph{Estimating reduced collision probabilities.}%
Both our robust clustering oracle and local reconstruction needs to invoke a procedure to estimate the \emph{reduced collision probability} of two random walks~\cite{KPS13:noise}. For a vertex $v$, an integer $t$ and a constant $\theta\in [0,1]$, we let $S_v^\theta=\{u:\pp_v^t(u)\leq \frac{1-\theta}{\sqrt{n}} \}$. For any two vertices $u,v$, the \emph{$\theta$-reduced collision probability} of $u,v$ is defined as 
$$\rcp_\theta(u,v)=\sum_{w\in S_u^\theta\cap S_v^\theta}\pp_u^t(w)\pp_v^t(w).$$

Observe that by definition of $\qq_v^t$-random walks, it holds that 
$$\rcp_0(u,v)\leq \sum_{w\in V}\pp_v^t(w)\cdot \pp_u^t(w)=\sum_{w\in V}\pp_v^t(w)\cdot \pp_w^t(u)=\qq_v^t(u).$$
The following lemma shows that under appropriate conditions, the reduced collision probability of two vertices can be well approximated in $\tilde{O}(\sqrt{n})$ time. %
\begin{lemma}[\cite{KPS13:noise}]\label{lemma:estimatercp}
Let $\theta<\frac{1}{2},\delta<1$ be two constant. Let $u,v$ be two vertices. There exists a procedure \textsc{EstimateRCP}($G,u,v,\theta,\delta,t$) that takes as input a $d$-bounded $n$-vertex graph $G$, vertices $u,v$, parameters $\theta,\delta$, and length parameter $t$, and satisfies the following properties: 
\begin{itemize}
\item[1)] It runs in time $O(\sqrt{n}t\log^2 n)$;
\item[2)] If $\pp_u^t(S_u^\theta)\geq 1/2,\pp_v^t(S_v^\theta)\geq 1/2$, then it aborts (without outputting an estimate) with probability at most $\exp(-\Theta(\sqrt{n}))$;
\item[3)] If it does not abort, then with probability at least $1-\frac{1}{n^4}$, it outputs an estimate $\rcp'(u,v)$ such that
$$\rcp_\theta(u,v)-\delta\max\{\rcp_\theta(u,v),\frac{1}{2n}\} \leq \rcp'(u,v) \leq \rcp_0(u,v) +\delta\max\{\rcp_0(u,v),\frac{1}{2n}\}.$$
\end{itemize}

\end{lemma}
For the sake of completeness, we give the description of the algorithm \textsc{EstimateRCP} in Appendix~\ref{sec:RCP_alg}.

\section{Robust Clustering Oracle}\label{sec:oracle}
In this section, we present our algorithm for constructing the robust clustering oracle and answering the clustering queries. In the \emph{preprocessing} (or \emph{learning})
phase, the algorithm learns the cores (corresponding to clusters in the clusterable part) of the graph. In the \emph{query} phase, the algorithm checks if the queried vertex $v$ belongs to any of the learned cores or not to decide if it is an outlier or not. If not, the algorithm will find the index $i$ corresponding to the cluster that $v$ belongs to.

We will use the reduced collision probability of random walks of length $t=\frac{960\log n}{\cst\phi^2}$ for some sufficiently small constant $\cst>0$. Such probabilities can be efficiently estimated by invoking the \textsc{EstimateRCP} procedure (see Section~\ref{sec:preliminaries}). The intuition is that for a typical vertex $u$ in a large cluster $C$, the uniform averaging walk of $t$ steps from $u$ will be close to the uniform distribution on $C$ (by Theorem~\ref{thm:rw_perturbed}), which implies that for almost all of vertices $v\in C$, their reduced collision probability  is at least $\frac{1-\kappa}{|C|}$. %

The learning phase of the algorithm is as follows. 
\begin{center}
\vspace{-1em}
	\begin{tabular}{|p{\textwidth}|}
		\hline
		\textbf{The preprocessing phase:} \textsc{LearnCore}$(G,d,k,\phi,\varepsilon)$ \\
		\hline
		\begin{tightenumerate}%
			\item Let $\theta_0$ and $\delta_0$ be two sufficiently small constant (say at most $\frac{1}{10^5}$).
			Let $\cst>0$ be a constant such that $\kappa= 100\cdot\delta_0^2$. If $\varepsilon>\frac{\phi\kappa^2}{100}$, then abort and output \textbf{fail}.
			
			\item Let $c>0$ be a sufficiently large constant. Let $\tau_j=3\sqrt{\frac{6\varepsilon}{\phi}}(1+\frac{\cst}{3})^j$ for $0\leq j\leq J$, where $J:=\arg \max_j\{\tau_j\leq 1\}$. (Note that $J=O(\frac{\log(\phi/\varepsilon)}{\kappa})$.) Let $t=\frac{960\log n}{\cst\phi^2}$.
			\item Sample a set $S$ of $\frac{c\cdot k^2\ln k\cdot \log n}{\sqrt{\varepsilon/\phi}}$ vertices uniformly at random.
			\item For any $u,v\in S$, run $\EstRCP(G,u,v,\theta_0, \delta_0, t)$. 
			If it does not abort then 	
			add an edge $(u,v)$ with weight $\rcp'(u,v)$ in the \emph{similarity graph} $H$ on vertex set $S$.
			\item Invoke \textsc{FindCore}($H, J$) (to find \emph{cores}).
			%
			%
			%
			%
			\vspace{-1em}
		\end{tightenumerate}\\
		\hline
	\end{tabular}
\end{center}

The subroutine \textsc{FindCore}($H, J$) is defined as follows.
\begin{center}
	\vspace{-1em}
	\begin{tabular}{|p{\textwidth}|}
		\hline
		\textsc{FindCore}$(H, J)$\\
		\hline
		\begin{tightenumerate}%
			\item Let $F=H$. Let $\mathcal{S}=\emptyset$. 
			\item For each $0\leq j\leq J$, we iteratively do the following: 
			\begin{tightenumerate}
			\item Let $F_j$ denote the subgraph of $F$ that consists of edges of weight at least $\frac{1-\kappa}{\tau_{j}}\frac{1}{n}$; 
			\item For each $v\in V(H)$: \hspace{3cm} $\triangleright$ according to the lexicographical order of vertices
	 \begin{tightenumerate}
	 	\item Let $N(v)$ denote the neighborhood of $v$ in $F_j$.
	 	\item Find a maximal clique $K$ from $v$ by sequentially visiting all the edges incident to vertex $v$ and all vertices $u\in N(v)$.  
	 	\item If a clique $K$ with $|K|\geq (1-\kappa)\tau_j |S|$ is found, then 1) add $K$ to $\mathcal{S}$, and
	 	2) remove all edges incident to $K$ from $F_j$ and $F$.
	 \end{tightenumerate} 
 			\end{tightenumerate}	
\item If $|\mathcal{S}|=0$ or $|\mathcal{S}|>k$, then output \textbf{fail}; otherwise, output all the disjoint cliques (called \emph{cores}), say $S_1,S_2,\cdots, S_h$, $h\leq k$, in $\mathcal{S}$.

			\vspace{-1em}
		\end{tightenumerate}\\
		\hline
	\end{tabular}
\end{center}

Note that by the above definition of cores, it holds that for any core $S_i$, there exists $j_i\in\{0,1,\cdots, J\}$ such that $\frac{|S_i|}{|S|}\geq (1-\cst)\tau_{j_i} $ and the edge weight in the clique $H[S_i]$ is at least $\frac{1-\cst}{\tau_{j_i}}\frac{1}{n}$. 

We need the following subroutine to answer clustering queries.
\begin{center}
\vspace{-1em}
	\begin{tabular}{|p{\textwidth}|}
		\hline
		\textsc{CheckCore}$(H,u)$\\
		\hline
		\begin{tightenumerate}%
			\item Let $\kappa, \theta_0$, $\delta_0$, $t$ and $\tau_j$ be the same numbers as specified in the learning phase.
			\item For any vertex $v\in S$, run \textsc{EstimateRCP}$(u,v,\theta_0, \delta_0, t)$. %
			\begin{tightenumerate}
				\item If there exists a \emph{unique} $i\leq h$ such that %
				$\rcp'(u,v)\geq \frac{1-\cst}{\tau_{j_i}}\frac{1}{n}$ for all $v\in S_i$, then return index $i$;
				\item Otherwise, return \textbf{Outlier}.
			\end{tightenumerate}
			\vspace{-1em}
		\end{tightenumerate}\\
		\hline
	\end{tabular}
\end{center}

Now we are ready to describe our algorithm for answering clustering queries.
\begin{center}
	\vspace{-0.5em}
	\begin{tabular}{|p{\textwidth}|}
		\hline
		\textbf{The query phase:} \\
		\hline
\textsc{IsOutlier}$(G,w)$:
\begin{tightenumerate}
	\item If the learning phase outputs \textbf{fail}, then return \textbf{Yes}.
	\item Otherwise, if \textsc{CheckCore}($H,w$) returns \textbf{Outlier}, then return \textbf{Yes}.
	\item Return \textbf{No}.
			\vspace{-1em}
\end{tightenumerate}
\\
\hline%
\textsc{WhichCluster}$(G,w)$:
		\begin{tightenumerate}%
			\item If \textsc{IsOutlier}($G,w$) returns \textbf{Yes}, return \textbf{Outlier}.
			\item 
			Otherwise, return \textsc{CheckCore}($H,w$). 
						\vspace{-1em}
		\end{tightenumerate}\\
		\hline
\textsc{SameCluster}$(G,x,y)$:
\begin{tightenumerate}%
\item Run \textsc{WhichCluster}$(G,x)$ and \textsc{WhichCluster}$(G,y)$.
\item If none of the above two queries return \textbf{Outlier} and the returned two indices are identical, then output \textbf{Yes}.
\item Otherwise, return \textbf{No}.
			\vspace{-1em}
\end{tightenumerate}
\\
\hline
	\end{tabular}
\end{center}

\subsection{The Analysis of Robust Clustering Oracle}
In the following, we show the performance guarantee of the above algorithm. We will use the local mixing property on noisy clusterable graphs as guaranteed in Theorem~\ref{thm:rw_perturbed}, whose proof is deferred to Section~\ref{sec:proof_localmixing}. Recall from the description of our algorithm that $\cst=100\cdot \delta_0^2$, which is a sufficiently small universal constant. 

If $\varepsilon> \frac{\phi\kappa^2}{100}$ (i.e., the noise is too much), then by our algorithm, the learning phase will output \textbf{fail}. Any queried vertex will be reported as \textbf{Outlier}.%

In the following, we assume that $\varepsilon \in [\Omega(\frac{\phi}{{n}}),   \frac{\phi\kappa^2}{100}]$ and we prove the statement of  Theorem~\ref{thm:oracle}. To do so, we first introduce the definition of strong vertices, which correspond to vertices in the clusterable part. %
\paragraph{Definition and properties of strong vertices.} 
Let $\phi\in (0,1), \varepsilon \in [\Omega(\frac{\phi}{{n}}),   \frac{\phi\kappa^2}{100}]$. Let $G$ be an $\varepsilon$-perturbation of a $(k,\phi,\frac{a_{\ref{thm:rw_perturbed}}\varepsilon\kappa^4\phi}{3k^3\log n})$-clusterable graph. %
Recall that $\pp_v^t$ and $\qq_v^t$ denote the distribution of the uniform average walk of length $t$ and the uniform average walk of length $t$ with two phases starting from $v$, respectively. In the algorithm, we invoke \textsc{EstimateRCP} with length parameter $t=\frac{960\log n}{\cst\phi^2}$. 

We let $\varepsilon':=\frac{6\varepsilon}{\phi}<\frac{\cst^2}{100}$. We introduce the following definition of strong vertex for the analysis, which was inspired by the corresponding definition for noisy expander graphs in~\cite{KPS13:noise}. The main difference here is that we carefully take the size of clusters into consideration. 
\begin{definition}
We call a vertex $v$ a \emph{strong} vertex with respect to a subset $C$ if $v\in C$,  $|C|\geq 3\sqrt{\varepsilon'} n$ and $\norm{\pp_v^{t}-\uni_{C}}_{\TV} \leq \cst.$
\end{definition}

Recall that $\theta_0$ is small sufficiently small constant,  $S_v^{\theta_0}=\{u:\pp_v^t(u)\leq (1-{\theta_0})/\sqrt{n} \}$ and that  $\rcp_{\theta_0}(u,v)=\sum_{w\in S_u^{\theta_0}\cap S_v^{\theta_0}}\pp_u^t(w)\pp_v^t(w)$ is the reduced collision probability of $u,v$ (see Section~\ref{sec:preliminaries}). We have the following properties of strong/weak vertices, which easily follows from the proof of Lemma 2 in \cite{KPS13:noise}. We present the proof in Appendix~\ref{sec:proof_rco} for the sake of completeness.
\begin{lemma}\label{lemma:strongvertices}
	If a vertex $u$ is strong with respect to a set $C$ with $|C|\geq 3\sqrt{\varepsilon'} n$, then (1) there can be at most $\sqrt{\cst} |C|$ vertices $v$ in $C$ with $\pp_u^t(v)\leq (1-\sqrt{\cst})/|C|$; (2) it holds that $\pp_u^t(S_u^{\theta_0})\geq 1/2$.

	Furthermore, if vertices $u,v$ are both strong with respect to a set $C$ with $|C|\geq 3\sqrt{\varepsilon'} n$, then we have that $\rcp_{\theta_0}(u,v) \geq (1-5\sqrt{\cst})/|C|$.%
\end{lemma}

\paragraph{The correctness of the robust clustering oracle.}

Now we show the correctness of the robust clustering oracle and bound the total number of vertices reported as outliers by the the algorithm.  Recall that we let $P_i:=\{u\in V: \textsc{WhichCluster}(u)=i\}$ with $1\leq i\leq h$ for some integer $h$, and $B:=\{u\in V: \textsc{IsOutlier}(u)=\textbf{Yes}\}$ denote the partition output by our algorithm.  
\begin{lemma}\label{lemma:reported_outliers}
Let $G$ be an $\varepsilon$-perturbation of a $(k,\phi,\frac{a_{\ref{thm:rw_perturbed}}\varepsilon\kappa^4\phi}{3k^3\log n})$-clusterable graph. Then there exists a partition ${D}_1,\cdots,{D}_{h'}, B'$ for some $h'\leq k$ (that is independent of the order of queries), such that 
\begin{itemize}
\setlength{\itemsep}{0pt}
\item if $\varepsilon\in [\Omega(\frac{\phi}{{n}}),\frac{\phi}{60k^2}]$, then
$h'\geq 1$ and each $D_i$ is a $(\frac{\phi}{2}, \frac{a_{\ref{thm:rw_perturbed}}\sqrt{\varepsilon}\kappa^4\phi^{1.5}}{3k^3\log n})$-cluster, for any $1\leq i\leq h'$; if $\varepsilon\in (\frac{\phi}{60k^2},1]$, then $h'=0$; and
\item with probability at least $1-\frac{1}{n}$, the partition $P_1,\cdots,P_{h},B$ output by the algorithm satisfies that $h'\leq h\leq k$ and  $\sum_{i=1}^{h'}|P_i\triangle D_i| + \sum_{i=h'+1}^{h}|P_i|+ |B|+|B'|\leq 40k\sqrt{\varepsilon/\phi}n$. 
\end{itemize}
In particular, the number of vertices reported as outliers is at most $40k\sqrt{\varepsilon/\phi}n$.
\end{lemma} 
\begin{proof}
We first note that if $\varepsilon> \frac{\phi}{60k^2}$, then we can simply take $B'=V$ (and thus $h'=0$) and then for any output partition of the algorithm, it holds that $\sum_{i=1}^{h'}|P_i\triangle D_i| + \sum_{i=h'+1}^{h}|P_i|+|B|+|B'| = 2n < 40k\sqrt{\varepsilon/\phi}n$.

Thus, in the following, we assume that $\varepsilon\leq \frac{\phi}{60k^2}$.
	
	Let $\phi_\myout=\frac{a_{\ref{thm:rw_perturbed}}\varepsilon\kappa^4\phi}{3k^3\log n}$. Let $G^*=(V,E^*)$ be a $(k,\phi,\phi_\myout)$-clusterable graph such that $G$ is an $\varepsilon$-perturbation of $G^*$. Let $C_1,\cdots, C_{\overline{h}}$ be the corresponding $(\overline{h},\phi,\phi_\myout)$-clustering of $G^*$ for some $\overline{h}\leq k$. %
	That is, for each $i\leq \overline{h}$, $\phi_G(C_i)\leq \phi_\myout$, and one can insert/delete at most $\varepsilon dn$ edges inside subgraphs $G[C_i]$ to make all $G[C_i]$ become $(\phi,\phi_\myout)$-clusters. %

	Now for each set $C_i$, we perform the following process on $G[C_i]$ recursively. %
	We start with $B_i:=\emptyset$ and $D_i:=C_i$. If $|B_i|\leq \frac{|C_i|}{2}$, and there exists a subset $M_i\subseteq D_i$ with $|M_i|\leq |D_i|/2$ and $\phi_{G[C_i]}(M_i)\leq \phi/2$, then we update $B_i=B_i\cup M_i$, and $D_i=D_i\setminus M_i$. We recurse until no such set $M_i$ can be found or $|B_i|>\frac{|C_i|}{2}$. Note that by our construction, the final set $B_i$ satisfies that $\phi_{G[C_i]}(B_i)\leq \phi/2$ and that $D_i$ has inner conductance at least $\phi/2$. Furthermore, it holds that $|B_i|\leq \frac{3}{4}|C_i|$, since right before the last update, we have that $|B_i'|\leq \frac{|C_i|}{2}$ and that the final cut $M'$ satisfies that $|M_i'|\leq \frac{1}{2}(|C_i-B_i'|)$, which gives that $|B_i|\leq \frac{1}{2}(|C_i-B_i'|)+|B_i'|\leq \frac{3}{4}|C_i|$. %
	
	Now we claim that $|\cup_i B_i|\leq \frac{6\varepsilon}{\phi} n$. Assume on the contrary that $|\cup_iB_i|> \frac{6\varepsilon}{\phi} n$, i.e., $\sum_i|B_i|>\frac{6\varepsilon}{\phi} n$. First, we note that in order to make $\phi(G[C_i])\geq \phi$, then we should add at least $\frac{\phi}{2}d\min\{|B_i|,|C_i-B_i|\}\geq \frac{\phi}{2} d\cdot \frac{1}{3}|B_i|=\frac{\phi}{6}d|B_i|$ edges, where the inequality follows from the fact that $|C_i-B_i|\geq \frac{1}{3}|B_i|$ which in turn is due to the fact that $|B_i|\leq \frac{3}{4}|C_i|$. Therefore, in order to make all $C_i$ have inner conductance at least $\phi$, we have to add at least $\sum_i \frac{\phi}{6}d|B_i|>\frac{\phi}{6}d\cdot \frac{6\varepsilon}{\phi} n=\varepsilon d n$ edges, which is a contradiction. 

We note that since $\varepsilon\leq \frac{\phi}{60k^2}$, then it holds that at least one $D_i$ has size at least $\frac{(1-(6\varepsilon/\phi))n}{k} \geq \frac{9n}{10k}\geq 3\sqrt{\frac{1}{10k^2}}n\geq 3\sqrt{\frac{6\varepsilon}{\phi}}n=3\sqrt{\varepsilon'} n$.	Now we apply Theorem~\ref{thm:rw_perturbed} on $G$ with error parameter $\varepsilon'=\frac{6\varepsilon}{\phi}<\frac{\cst^2}{100}$, $\ccss = \frac{\cst}{2}$, sets $C_i=D_i\cup B_i$, $1\leq i\leq \overline{h}$ such that $\phi(G[D_i])\geq \frac{\phi}{2}$, to obtain that for each $D_i$ with $|D_i|\geq 3\sqrt{\varepsilon'}n$, there exists a subset $\widehat{D}_i\subseteq D_i$ such that $|\widehat{D}_i|\geq (1-4\sqrt{\varepsilon'})|D_i|$ and for any $v\in \widehat{D}_i$, and $t= \frac{960\log n}{\cst\phi^2}$,
	$$\norm{\pp_{v}^{t}- \uni_{C_i}}_{\TV}%
	\leq \sqrt{\varepsilon'}+\frac{\cst}{2}\leq \cst.$$

	This further implies that all vertices in $\widehat{D}_i$ are strong with respect to $C_i$, as $|C_i|\geq |D_i|\geq 3\sqrt{\varepsilon'}n$. We also note that for each $D_i$ with $|D_i|\geq 3\sqrt{\varepsilon'}n$, it holds that $\phi_G(D_i)\leq \frac{\phi_\myout \cdot dn}{3\sqrt{\varepsilon'}\cdot dn}\leq \frac{a_{\ref{thm:rw_perturbed}}\sqrt{\varepsilon}\kappa^4\phi^{1.5}}{3k^3\log n}$. Now we order $D_i$ such that $|D_1|\geq \cdots \geq |D_{\overline{h}}|$ (breaking ties arbitrarily). Let $h'$ be the largest index with $|D_{h'}|\geq 3\sqrt{\varepsilon'}n$. Note that $h'\geq 1$. We define the partition $D_1,\cdots,D_{h'}, B':=V\setminus(\cup_{i\leq h'} D_i) $. By definition, it holds that for each $1\leq i\leq h'$, $|D_i|\geq 3\sqrt{\varepsilon'}n$ and $\phi(G[D_i])\geq \frac{\phi}{2}, \phi_G(D_i)\leq \frac{\phi_\myout \cdot dn}{3\sqrt{\varepsilon'}\cdot dn}\leq \frac{a_{\ref{thm:rw_perturbed}}\sqrt{\varepsilon}\kappa^4\phi^{1.5}}{3k^3\log n}$. Note that the partition $D_1,\cdots,D_{h'}, B'$ only depends on $G$. It holds that $|B'|=|\sum_{i}B_i| + |\cup_{i:|D_i|<3\sqrt{\varepsilon'}n} D_i|\leq (\varepsilon'+ 3k\sqrt{\varepsilon'}) n$.
	
	We further define $D_g:=\cup_{1\leq i\leq h'} \widehat{D}_i$.

 	Now we show the following claim.
	\begin{claim}\label{claim:strongaccept}
		With probability at least $1-\frac{1}{n}$, for all vertices $v$ in $D_g$, \textsc{WhichCluster}($v$) will output a unique index $\sigma(i)$ if vertex $v\in \widehat{D}_i$ for some injection $\sigma: [h']\rightarrow [k]$. 
\end{claim}

	Note that the statement of the lemma will then follow from the above claim: Let $h\leq k$ be the largest index output by the algorithm, and let $B,P_{\sigma(1)},\cdots,P_{\sigma(h')},P_{j}$, for $j\in[h]\setminus\{\sigma(1),\cdots,\sigma(h')\}$ be the partition output by the algorithm. Then by Claim~\ref{claim:strongaccept}, all vertices in $D_g$ will be correctly partitioned and 
	\begin{eqnarray*}
&&\sum_{i=1}^{h'}|P_{\sigma(i)}\triangle D_i| + \sum_{j\in [h]\setminus\{\sigma(1),\cdots,\sigma(h')\}}|P_j| + |B|+|B'|\\
&\leq& 2\cdot (|\cup_{i:|D_i|<3\sqrt{\varepsilon'}n} D_i|+|\cup_{i:|D_i|\geq 3\sqrt{\varepsilon'}n}(D_i\setminus \widehat{D}_i)|+|\cup_iB_i|) \\
&\leq& 2\cdot (3k\sqrt{\varepsilon'}+4\sqrt{\varepsilon'}+\varepsilon')n \\
&\leq& 16k\sqrt{\varepsilon'} n \leq 40k\sqrt{\varepsilon/\phi}n.
\end{eqnarray*}

	Re-arranging the order of sets $D_1,\cdots,D_i$ will complete the proof of the Lemma. Now we prove the claim. 
	
	\begin{proof}[Proof of Claim~\ref{claim:strongaccept}]
		By the previous analysis, we have that for each $i$ such that $|D_i|\geq 3\sqrt{\varepsilon'}n$, the number of vertices in $C_i$ that are \emph{not} strong (with respect to $C_i$) is at most 
		$$|D_i\setminus \widehat{D}_i| + |B_i|\leq 4\sqrt{\varepsilon'}|D_i|+\varepsilon' n\leq 5\sqrt{\varepsilon'} |D_i|\leq 5\sqrt{\varepsilon'}|C_i|\leq \frac{\cst}{2}|C_i|. $$
		
		That is, for each $i$ such that $|C_i|\geq |D_i|\geq 3\sqrt{\varepsilon'}n$, at least $(1-\frac{\cst}{2})$ fraction of vertices in $C_i$ are strong (with respect to $C_i$). 
		
		Now let us consider the sample set $S$. Recall that $|S|=\frac{c\cdot\log n}{\sqrt{\varepsilon/\phi}}=\Omega(\frac{\log n}{\sqrt{\varepsilon'}})$ for some large constant $c>0$. Let $T_i=S\cap C_i$ and let $S_i'\subset T_i$ denote the set of vertices in $T_i$ that are strong with respect to $C_i$. By Chernoff bound, we have that with probability at least $1-1/n^4$, for any $i$ such that $|C_i|\geq|D_i|\geq 3\sqrt{\varepsilon'}n$,
		\begin{eqnarray*}
			(1-\frac{\cst}{2})\frac{|C_i|}{n}\cdot |S| &\leq&	|T_i|\leq (1+\frac{\cst}{2})\frac{|C_i|}{n}\cdot |S|,\\
			(1-\cst)\frac{|C_i|}{n}\cdot |S|<(1-\frac{\cst}{2})(1-\frac{\cst}{2})\frac{|C_i|}{n}\cdot |S|&\leq& |{S}_i'|\leq (1+\frac{\cst}{2})\frac{|C_i|}{n}\cdot |S|.
		\end{eqnarray*}
		In the following, we will condition on event that the above two inequalities hold.
		
		Now recall that $\tau_j=3\sqrt{\varepsilon'}(1+\frac{\cst}{2})^j$, for $0\leq j\leq J$, where $J=O(\frac{\log(\phi/\varepsilon)}{\kappa})$ is the maximum integer $j$ such that $\tau_j\leq 1$. Let $j_i$ denote the index such that $|C_i|\in [\tau_{j_i} n, \tau_{j_i+1}n)$. Thus, $|{S}_i'|\geq (1-\cst)\tau_{j_i}|S|$. 
		
		Let $v,u$ be two vertices in ${S}_i'$. By Lemma~\ref{lemma:strongvertices}, we have that $\pp_u^t(S_u^{\theta_0})\geq 1/2, \pp_v^t(S_v^{\theta_0})\geq 1/2$, and $\rcp_{\theta_0}(u,v)\geq (1-5\sqrt{\cst})/|C_i|$. By the assumption that $\kappa=100\cdot\delta_0^2$ and Lemma~\ref{lemma:estimatercp}, we obtain that with probability at least $1-\frac{1}{n^4}-\exp(-\Theta(\sqrt{n}))$, \textsc{EstimateRCP}$(G,u,v,\theta_0, \delta_0, t)$ will output a value $\rcp'(u,v)$ that is at least $(1-5\sqrt{\cst})(1-\frac{\sqrt{\cst}}{10})/|C_i|>\frac{1-\cst/2}{|C_i|}>\frac{1-\cst/2}{\tau_{j_i+1}n}\geq \frac{1-\cst}{\tau_{j_i}n}$. That is, with probability at least $1-\frac{1}{n^3}$, in the similarity graph $H$, the induced subgraph $H[{S}_i']$ will form a complete graph with at least $(1-\cst)\tau_{j_i} |S|$ vertices such that for each pair $u,v\in {S}_i'$, $\rcp'(u,v)\geq\frac{1-\cst}{\tau_{j_i}n}$. Therefore, in our sample, the set $S_i'$ will be recognized as a subgraph of a core (corresponding to $C_i$), which is a maximal clique with edge weight at least $\frac{1-\cst}{\tau_{j_i}n}$. %

		Now once a vertex $v\in \widehat{D}_i$ is queried (for checking if it is outlier or not), then by using similar argument as above, we can guarantee that with probability at least $1-\frac{1}{n^3}$, for all $u\in S_i'$, the \textsc{EstimateRCP} will output $\rcp'(u,v)$ satisfying that $\rcp'(u,v)\geq \frac{1-\kappa}{\tau_{j_i}n}$. Thus, the algorithm will detect the core (corresponding to $C_i$) for $v$. Furthermore, for any vertex $v$ that is strong with respect to $C_i$, it holds that for any $C_j$ with $j\neq i$, there can be at most $\kappa |C_j|$ vertices $u\in C_j$ with $\rcp_{\theta_0}(u,v)>\frac{1-5\sqrt{\kappa}}{|C_j|}$, this is true since the total probability mass on $C_j$ of the random walk distribution from $v$ is at most $\kappa$. This ensures that there will be a unique core corresponding to $v$. Let $\sigma:[h']\rightarrow[h']$ denote the corresponding bijection between $\{C_1,\cdots,C_{h'}\}$ and the cores $\{S_1,\cdots, S_{h'}\}$ found by the algorithm. By union bound, we have that with probability at least $1-\frac{1}{n}$, for each strong vertex $v\in S_i'$, the algorithm will answer the corresponding index $\sigma(i)$ to the query \textsc{WhichCluster}($v$). %
	\end{proof}
\end{proof}

\paragraph{Running time and query complexity.} Note that in the learning phase, we need to invoke the procedure \textsc{EstimateRCP} for $|S|\times |S|=O(\frac{k^4\ln^2 k \cdot \phi\log^2n}{\varepsilon})$ times, and each invocation takes time $O(\sqrt{n}t\log^2n)$, which in total takes time $O(\sqrt{n}\frac{\log n}{\phi^2}\cdot \frac{k^4\ln^2 k\cdot \phi\log^2n}{\varepsilon})=O(\sqrt{n}\frac{k^4\ln^2 k\cdot\log^3 n}{\phi\varepsilon})$. Finding the cores in the similarly graph can be implemented by a simple greedy algorithm？？？？, which can be implemented in $O(\poly(|S|))=O(\poly(\frac{k\cdot\phi \log n}{\varepsilon}))$ time. %
Thus, the query complexity and running time in the learning phase is dominated by $O(\sqrt{n}\cdot \poly(\frac{k\cdot \log n}{\varepsilon \phi}))$, which, by similar arguments, also upper bounds the query complexity and running time on each query vertex $w$ in the query phase.

\paragraph{Remark.} From Lemma \ref{lemma:reported_outliers} and its proof, we note that in order to guarantee that $h'\geq1$, i.e., there exists at least one good cluster $D_i$, we need to set $\varepsilon =O(\frac{\phi}{k^2})$ (so that $(1-\varepsilon')n/k\geq 3\sqrt{\varepsilon'}n$). Thus our algorithm has non-trivial guarantee only if the adversary does not perturb the graph too much. Suppose that there are $h\leq k$ ground-truth clusters $C_1,\cdots,C_h$ and the adversary perturbs an $\varepsilon$-fraction on intra-cluster edges. In order to recover for \emph{each} $C_i$, a subset $P_i$ that is close to $D_i\subseteq C_i$, then we need to require that $\min_{i\in h} |D_i|\geq 3\sqrt{\varepsilon'}n$, which can be satisfied if $\varepsilon = O(\phi \cdot (\frac{\min_{i\in h} |D_i|}{n})^2)$. 

We further remark that our algorithm can only be able to (partially) recover the large clusters, say of size at least $\Omega(\varepsilon n)$. This is the case as for any small cluster (of size $o(\varepsilon n)$), it can be completely hidden or destroyed by the adversary. Currently, our analysis shows that our algorithm can recover the cluster of size $\Omega(\sqrt{\frac{\varepsilon}{\phi}}n)$. It will be an interesting question to design a robust clustering oracle that can recover smaller clusters (i.e., of size in the range $[\Omega(\varepsilon n), o(\sqrt{\frac{\varepsilon}{\phi}}n)]$). %

\section{The Local Reconstruction Algorithm}\label{sec:algorithm}
In this section, we present our reconstruction algorithm, which will be built upon our robust clustering oracle algorithm in Section~\ref{sec:oracle} and consists of two phases: the \emph{learning} phase, that learns the cores (corresponding to clusters in the clusterable part) of the graph, and the \emph{query} phase, which first checks if the queried vertex belongs to any of the learned cores or not, and then output its neighbors in the amended clusterable graph accordingly. We need the following tool of explicit construction of expanders. %

\paragraph{Explicit expanders.} For any vertex set $V=[n]$, we let $G_{\exp}=(V, E_{\exp})$ denote a graph on $V$ with maximum degree at most $16$ such that for any set $S$ in $G_{\exp}$ with $|S|\leq n/2$, it holds that $|E_{\exp}(S,V\setminus S)|\geq \eta |S|$, for some constant $\eta > 0$. It is known (see e.g.,~Lemma 6 in~\cite{KPS13:noise} which builds upon~\cite{GG81:explicit}) that such an expander $G_{\exp}$ also exists and can be \emph{explicitly} constructed in the sense that for any specified vertex $v$, one can find all neighbors of $v$ in $G_{\exp}$ in $\poly(\log n)$ time. 

In the following, given a graph $G$, we let $G_{\exp}$ denote an explicit expander graphs on the same vertex set as $G$. We call vertices $G$ or $G_{\exp}$-neighbors of a vertex $v$, depending on the graph under consideration. 
\vspace{-1em}

\begin{center}
	\begin{tabular}{|p{\textwidth}|}
		\hline
		\textbf{Local reconstruction:}  \\
		\hline
		\begin{tightenumerate}%
			\item Run \textsc{LearnCore}$(G,d,k,\phi,\varepsilon)$.
			\item For each query \textsc{NewNeighbors}$(G,w)$:
			\begin{tightenumerate}
	\item If the learning phase outputs \textbf{fail}, then output all $G_{\exp}$-neighbors and $G$-neighbors of $w$.
\item %
Otherwise, run \textsc{CheckCore}($H,w$). 
\begin{tightenumerate}
	\item If $w$ is reported as \textbf{Outlier}, then add all the $G_{\exp}$-edges $(w,u)$ incident to $w$; 
	\item Otherwise, for each vertex $u$ that is a $G_{\exp}$-neighbor of $w$, run \textsc{CheckCore}($H,u$). If $u$ is reported as \textbf{Outlier}, then add edge $(w,x)$.
\end{tightenumerate}
Output all neighbors added to $w$ and all $G$-neighbors of $w$.
\vspace{-1em}
			\end{tightenumerate}				
		\end{tightenumerate}\\
		\hline
	\end{tabular}
\end{center}

Note that the algorithm should be implemented by first taking as input a random seed $s$, which is fixed once for all (and used for sampling vertices in the learning phase and performing random walks), and then on any query vertex $v$, \emph{deterministically} outputting the neighborhood of $v$ in the graph $G'$. %
By construction, if an edge $(u,v)$ is added, then on query vertex $u$, $v$ will be output as a neighbor of $u$ and vice versa. Therefore, the algorithm is independent of the order of queries and the answer will be globally consistent. %
\subsection{Analysis of the Local Reconstruction Algorithm}\label{sec:analysis}
In the following, we show the performance guarantee of the above algorithm and prove Theorem~\ref{thm:main}. We first note that the running time and query complexity can be analyzed in the same way as in the proof Theorem~\ref{thm:oracle}. %

It follows from the definition of $G_{\exp}$ that the maximum degree of $G'$ is bounded by $d+16$, as $G_{\exp}$ has maximum degree at most $16$ and for each vertex $u$ that is found to be an outlier, we will add all of its $G_{\exp}$-neighbors to $u$. 

Recall from the description of our algorithm that $\cst>0$ is a sufficiently small universal constant. If $\varepsilon> \frac{\phi\kappa^2}{100}$ (i.e., the noise is too much), then by our algorithm, the learning phase will output \textbf{fail}. Furthermore, on query any vertex $u$, the query phase will output all of its $G$ and $G_{\exp}$ neighbors of $u$. Thus, $G'$ is a complete hybridization of $G$ and $G_{\exp}$. Note that for any set $S\subset V$, $|E'(S,\bar{S})|\geq |E_{\exp}(S,\bar{S})|$, where $E'$ and $E_{\exp}$ denote the set of edges in $G'$ and $G_{\exp}$ respectively. Thus, it holds that if $|S|\leq \frac{n}{2}$,  $\phi_{G'}(S)=\frac{|E_{\exp}(S,\bar{S})|}{d|S|}\geq \frac{\eta}{d}$, where we used the fact that for any set $S$ with $|S|\leq \frac{n}{2}$ in $G_{\exp}$, $|E_{\exp}(S,\bar{S})|\geq \eta |S|$. Therefore, the resulting graph $G'$ is $(1,\frac{\eta}{d},0)$-clusterable. Furthermore, the number of edges added to $G$ is at most $16n/2=8n= O(\min\{1,k\sqrt{\varepsilon/\phi}\}\cdot n)$ as $\varepsilon>\frac{\phi\kappa^2}{100}$. Thus, in this case, the statement of our theorem holds.

In the following, we prove the rest properties as listed in Theorem~\ref{thm:main} for the more interesting case that $\varepsilon \in [\Omega(\frac{\phi}{{n}}),   \frac{\phi\kappa^2}{100}]$. %

In this case, the description of the local reconstruction algorithm, the number of added edges is $16$ times the number of vertices that are reported as outliers, and thus by Lemma~\ref{lemma:reported_outliers}, is at most $16\times 40k\sqrt{\frac{\varepsilon}{\phi}}n = 640k\sqrt{\frac{\varepsilon}{\phi}}n$. Now we analyze the cluster structure of the resulting graph. 

\paragraph{Definition and property of weak vertices.} Let $\varepsilon':=\frac{6\varepsilon}{\phi}<\frac{\cst^2}{100}$. We introduce the following definitions of weak vertex for the analysis, which was inspired by the corresponding definitions for noisy expander graphs in~\cite{KPS13:noise}. The main difference here is that we carefully take the size of clusters into consideration. 
\begin{definition}	
We call a vertex $v$ \emph{weak} vertex, if for any subset $A$ with $|A|\geq \frac{2\varepsilon'}{3} n$, it holds that $\norm{\bb_v^{t}-\uni_{A}}_{\TV} \geq 1/4.$
\end{definition}

In order to analyze the cluster structure of the resulting graph $G'$, we need the following property of weak vertices.
\begin{lemma}\label{lemma:weak}
With probability at least $1-n^{-3}$, it holds that for any weak vertex $u$, the algorithm will report $u$ as an outlier.
\end{lemma}
\begin{proof}
	We first show that if $u$ is weak, then for any subset $A$ with $|A|\geq \frac{2\varepsilon'}{3} n$ vertices, at most $7/8|A|$ vertices $v$ in $A$ satisfy $\qq_u^t(v)\geq \frac{7/8}{|A|}$. This is true since otherwise, there will be more than $7/8|A|$ vertices $v$ satisfy $\qq_u(v)\geq \frac{7/8}{|A|}$. If we let $A_1\subseteq A$ (resp. $A_2\subseteq A$) denote the set of vertices $v$ in $A$ such that $\qq_u^t(v)\leq \frac{1}{|A|}$ (resp. $\qq_u^t(v)>\frac{1}{|A|}$), then
	\begin{eqnarray*}
		\norm{\qq_u^t-\uni_{A}}_{TV}&=&\frac{1}{2}\left(\sum_{v\in A_1}(\frac{1}{|A|}-\qq_u^t(v)) + \sum_{v\in A_2}(\qq_u^t(v)-\frac{1}{|A|}) + \sum_{v\in V\setminus A}\qq_u^t(v)\right)\\
		&=&\sum_{v\in A_1}(\frac{1}{|A|}-\qq_u^t(v))\\
		&<& (1-7/8)|A|\cdot \frac{1}{|A|} + |A|\cdot \frac{1-7/8}{|A|} =2(1-7/8)<\frac{1}{4},
	\end{eqnarray*}
	which is a contradiction. By the definitions of reduced collision probability $\rcp_{\theta_0}(u,v)$ and relations of $\pp_u^t$ and $\qq_u^t$, we have that $\rcp_0(u,v)\leq \qq_u^t(v)$, and thus there can be at most $\frac{7}{8}|A|$ vertices $v$ in $A$ with $\rcp_0(u,v)\geq \frac{7}{8|A|}$. Note that this property holds for all sets $A$ with $|A|\geq \frac{2\varepsilon'}{3} n$.
	
	For each $0\leq j\leq J$, we let $T_j$ denote the set of vertices $v$ such that $\rcp_0(u,v)\geq \frac{7}{8\tau_{j}n}$. Recall that $\tau_j=3\sqrt{\varepsilon'}(1+\frac{\kappa}{2})^j$, for $0\leq j\leq J$.

	If $|T_j|\geq \tau_j n>\frac{2\varepsilon'}{3}n$, then for all vertices $v\in T_j$, it holds that $\rcp_0(u,v)\geq \frac{7}{8\tau_j n}\geq \frac{7}{8|T_j|}$, which is a contradiction. If $\frac{7\tau_j}{8}n< |T_j|<\tau_j n$, then we can add arbitrarily at most $\frac{1}{8}\tau_jn$ vertices to $T_j$ to obtain a set $A$ such that $|A|= \tau_j n >\frac{2\varepsilon'}{3}n$, and for at least $\frac{|T_i|}{|A|}> \frac{7}{8}$ fraction of vertices $v$ in $A$, it holds that $\rcp_0(u,v)\geq \frac{7}{8\tau_{j}n}= \frac{7}{8|A|}$, which is a contradiction. Therefore, it must hold that $|T_j|\leq \frac{7\tau_jn}{8}$. %

	That is, for the weak vertex $u$, it holds that for each $0\leq j\leq J$, there will be at most $\frac{7}{8}\tau_jn$ vertices $v$ with $\rcp_0(u,v)\geq \frac{7}{8\tau_j n}$. Thus, there will be at least $(1-\frac{7}{8}\tau_j)n$ vertices $v$ with $\rcp_0(u,v)\leq \frac{7}{8\tau_j n}$. We can further guarantee that with probability at least $1-1/n^2$, for any such pair $u,v$, the procedure \textsc{EstimateRCP} (with parameter $\delta\leq \frac{\sqrt{\cst}}{10}$) either aborts or outputs an estimate $\rcp'(u,v)\leq (1+\frac{\sqrt{\kappa}}{10})\frac{7}{8\tau_j n}\leq \frac{8}{9\tau_j n}$, for any $0\leq j\leq J$. Finally, with probability at least $1-\frac{2}{n^2}$, in our sample set $S$, at least $(1-\frac{7}{8}\tau_j)$ fraction of vertices $v$ satisfy that $\rcp'(u,v)\leq \frac{8}{9\tau_j n}$, or equivalently, less than $\frac{7}{8}\tau_j$ fraction of vertices $v$ satisfy that $\rcp'(u,v)\geq \frac{8}{9\tau_j n}$. This implies that our algorithm will report $u$ as an outlier.    
\end{proof}

\paragraph{Cluster structure of $G'$.} Now we are ready to show that the resulting graph $G'$ from our local reconstruction algorithm can be partitioned into at most $k$ parts, each of which has relatively large inner conductance.

\begin{lemma}\label{lemma:conductance}
Let $\phi^*=\frac{a_{\ref{lemma:conductance}}\varepsilon\phi}{k^4\log n}$ for some sufficiently small constant $a_{\ref{lemma:conductance}}$. %
If $G$ is an $\varepsilon$-perturbation of a $(k,\phi,\frac{a_{\ref{thm:rw_perturbed}}\varepsilon\kappa^4\phi}{3k^3\log n})$-clusterable graph, then the resulting graph	$G'$ from the local reconstruction algorithm is $(k,\phi^*,1)$-clusterable.%
\end{lemma}
\begin{proof}
For analysis, we perform the following procedure on the input graph $G$. Let $\gamma=\frac{\varepsilon'}{3}=\frac{2\varepsilon}{\phi}$. We start with the set $U:=V$ and a partitioning $\PP:=\{V\}$ of $G$. Then if there exists a set $U\in \PP$ and $S\subseteq U$ such that $\gamma n\leq |S|\leq \frac{|U|}{2}$ and $\phi_G(S)\leq \phi_\myout:=\frac{a_{\ref{thm:rw_perturbed}}\varepsilon\kappa^4\phi}{3k^3\log n}$, then we set $\PP=(\PP \setminus \{U\}) \cup\{S, U\setminus S\}$. We repeat until no such $S$ can be found. Let $\PP=\{C_1,\cdots, C_h \}$ denote the final partitioning of $V$.
	
	Note that for any $C_i$, if $U$ is the subset that contains $C_i$ and is then split into $C_i$ and $U\setminus C_i$, then $|U|\geq 2\gamma n$ and thus $|C_i|\geq \gamma n$ and $|U\setminus C_i|\geq \frac{|U|}{2}\geq \gamma n$ by the construction. This implies that at the end of the above procedure, it holds that $\min_i|C_i|\geq \gamma n$. 
	
	We further note that $|\mathcal{P}|=h\leq k$. This is true since otherwise, in order to make $G$ become a $(k,\phi,\phi_\myout)$-clusterable graph, one has to patch up at least one set $C_i$ to other parts, that is, we need to add at least $\frac{3\phi}{4}\cdot d \min_i\{|C_i|\}\geq \frac{3\phi}{4}\cdot d\cdot \frac{2\varepsilon}{\phi} n>\varepsilon d n$ edges, which is a contradiction to the assumption that $G$ is an $\varepsilon$-perturbation of a $(k,\phi,\phi_\myout)$-clusterable graph.

	Now let us consider the partition $\mathcal{P}$ in the constructed graph $G'$. Observe that by the description of our algorithm, for any set $S$ of vertices $|E'(S,\bar{S})|\geq |E(S,\bar{S})|$, where $E'$ and $E$ denote the set of edges in $G'$ and $G$ respectively. %
	In particular, Lemma~\ref{lemma:weak} implies that the set of $G'$-neighbors of any weak vertex $u$ is a superset of the set of $G$-neighbors of $u$, as $u$ will be reported as an outlier by the algorithm and the $G_{\exp}$-neighbors of $u$ will be added to $G'$.

	We have the following claim. 
	\begin{claim}
		In the graph $G'$, for each $C_i$, and any subset $S\subset C_i$ with $|S|\leq \frac{|C_i|}{2}$, it holds that $\phi_{G'}(S)\geq k\phi^*$.
	\end{claim} 
	\begin{proof}
If $\gamma n \leq |S|\leq \frac{|C_i|}{2}$, then by our construction of $C_i$, we have that $\phi_G(S)\geq \phi_\myout$. Thus, $\phi_{G'}(S)=\frac{|E'(S,V\setminus S)|}{d|S|}\geq \frac{|E(S,V\setminus S)|}{d|S|}\geq \phi_\myout\geq k\phi^*$. Now let us consider the case that $|S|\leq \gamma n$. 

If there are less than $(1-\frac{\eta}{2})$ fraction of vertices in $S$ are weak, then we show that $\phi_G(S)\geq k\phi^*$. Suppose this is not the case, that is, $\phi_G(S)<\frac{a_{\ref{lemma:conductance}}\varepsilon\phi}{k^3\log n}\leq \frac{\eta}{16t}$, if we set $a_{\ref{lemma:conductance}}$ to be a sufficiently small constant. By the proof of Theorem 4 in~\cite{KPS13:noise} (which in turn is based on the proof of Lemma 4.7 in~\cite{CS10:expansion}), we know that for at least $(1-\eta/2)$ fraction of vertices $u$ in $S$, the probability that a $\qq_u^{t}$-random walk that starts at $u$ will end up in $\bar{S}$ is at most $1/4$. Now let $A$ be any set with $|A|\geq \frac{2\varepsilon'}{3} n$. Since $|S|\leq \frac{\varepsilon'}{3} n$, it holds that $|A\setminus S|\geq \frac{1}{2}|A|$. Thus, we have that $\uni_A(A\setminus S)\geq \frac{1}{2}$. This gives that $\norm{\qq_u^t-\uni_A}_\TV\geq \frac{1}{4}$, which implies that such a vertex $u$ is weak. Thus, $S$ contains at least $(1-\eta/2)$ fraction of weak vertices, which is a contradiction. This implies that $\phi_{G'}(S)\geq \phi_G(S)\geq k\phi^*$.

If there are more than $(1-\eta/2)$ fraction of weak vertices, denoted by $W$, in $S$, then the number of $G_{\exp}$-neighbors of $W$ in $G_{\exp}$ is at least $\eta|W|$. Since all these $G_{\exp}$-neighbors are also in $G'$, we have that the number of vertices outside of $S$ is at least $\eta|W|-|S\setminus W|\geq \eta(1-\eta/2)|S|-\eta/2|S|\geq \frac{\eta}{6}|S|$. Since we add all the edges in $G_{\exp}$ that are incident to $W$ to $G'$, we have that the number of edges crossing $S$ in $G'$ is at least $\frac{\eta}{6}|S|$, and thus $\phi_{G'}(S)\geq \frac{\eta}{6d}\geq k\phi^*$.
\end{proof}
	
Now based on the partition $\mathcal{P}=\{C_1,\cdots,C_h \}$ as constructed above, we find a new partition of $G'$ such that each part has large inner conductance. %
We start with the partition $\mathcal{P}=\{C_1,\cdots,C_h \}$ as constructed above and perform the following operations. If there exist $i,j\leq h$, $S\subseteq C_i$ satisfies that $i\neq j$, $|S|\leq \frac{|C_i|}{2}$ and that $|E'(S,C_i\setminus S)| < |E'(S,C_j)|$, then we set $C_i:=C_i\setminus S$ and $C_j:=T_j\cup S$. We repeat until the condition is violated. 

Note that the above process always terminates in a finite number of steps since the number of crossing edges, i.e., $\sum_{i\neq j}|E'(C_i,C_j)|$, always decreases in each iteration. Furthermore, we observe that at the end of the process, for any $1\leq i\leq h$, and any set $S\subseteq C_i$ with $|S|\leq \frac{|C_i|}{2}$, $|E'(S,C_i\setminus S)| \geq \frac{|E'(S,V\setminus S)|}{k}$. Therefore, $\phi_{G'[C_i]}(S)\geq \frac{1}{k}\phi_{G'}(S)\geq \phi^*$. This implies that for each $i$, $\phi(G'[C_i])\geq \phi^*$. 
\end{proof}

\section{Local Mixing Property of Random Walks on Noisy Clusterable Graphs: Proof of Theorem~\ref{thm:rw_perturbed}}\label{sec:proof_localmixing}
In this section, we give the proof of Theorem~\ref{thm:rw_perturbed}. To do so, we first give a property of random walks on clusterable graphs (without noise).%

\subsection{Local Mixing Property of Random Walks on Clusterable Graphs} 
We will first prove a mixing property of random walks on a clusterable graph, which says that in a clusterable graph, a random walk of appropriate length starting from a typical vertex of a large cluster will mix well inside the corresponding cluster. By a simple reduction (see Section~\ref{sec:preliminaries}), it suffices to consider a corresponding weighted $d$-regular graph for any $d$-bounded graph. %

\begin{theorem}\label{thm:rw_clusterable}
	Let $0<\alpha,\beta,\xi \leq 1$. Let $\phi_\myout\leq a_{\ref{thm:rw_clusterable}}\frac{\xi{\alpha\beta}\phi_\myin^2}{k^3\log n}$ for some sufficiently small constant $a_{\ref{thm:rw_clusterable}}>0$. Let $G$ be a weighted $d$-regular and $(k,\phi_\myin,\phi_\myout)$-clusterable graph with underlying clusters $C_1,\cdots, C_h$ for some $h\leq k$. Then for each $C_i$ with $|C_i|\geq \alpha n$, there exists a subset $C'_i\subseteq C_i$ such that $|C_i'|\geq (1-\beta)|C_i|$, and for any $v\in C'_i$, and $t= \frac{20\log n}{\phi_\myin^2}$, it holds that
	$$\norm{\p_v^{t}-\uni_{C_i}}_{\TV}\leq \xi.$$
\end{theorem}
We remark that \cite{ST13:local} and \cite{AOPT16:local} gave analysis for upper bounding the probability that a random walk of length $t$ from a typical vertex $v$ in a set $S$ with small conductance will escape the set $S$, and lower bounding the probability that the walk from $v$ of length $t$ stays inside $S$, respectively. It is unclear if one can use their analysis to prove the above theorem. In the following, we prove Theorem~\ref{thm:rw_clusterable} by using %
some strong spectral property of clusterable graphs, i.e., the spectral gap between $\lambda_{h+1}$ and $\lambda_h$ for some $h\leq k$, and the closeness of the space spanned by the first $h$ eigenvectors and the space spanned by the indicator vectors of clusters. More precisely, %
we need the following tools. 

\begin{lemma}[Lemma 5.2 in~\cite{CPS15:cluster} and Lemma 10 in~\cite{CKKMP18:cluster}]\label{lemma:spectral_clusterable}
	Let $G$ be a weighted $d$-regular and $(k,\phi_\myin,\phi_\myout)$-clusterable graph with underlying clusters $C_1,\cdots, C_h$ for some $h\leq k$. Then $\lambda_h\leq 2\phi_\myout$ and $\lambda_{h+1}\geq \frac{\phi_\myin^2}{2}$. %
\end{lemma}

\begin{fact}\label{fact:sum_1}
	It holds that  $\norm{\1_v}_2^2=\sum_{j=1}^n\vect{v}_j(v)^2=1$, for any $v\in V$.
\end{fact}

The following is a direct corollary of a structural result due to \cite{PSZ17:partition} that relates the first $k$ eigenvectors of the Laplacian to the normalized indicator vectors of some $k$-partition of the graph. Recall that $\vect{v}_i$ is the eigenvector corresponding to the $i$-th smallest eigenvalue of the Laplacian of $G$.
\begin{theorem}\label{thm:PSZ_structure}
	Let $\phi_\myout\leq a_{\ref{thm:PSZ_structure}}\phi_\myin^2/k^2$ for sufficiently small constant $a_{\ref{thm:PSZ_structure}}>0$. Let $G$ be a weighted $d$-regular and $(k,\phi_\myin,\phi_\myout)$-clusterable graph with underlying $(\phi_\myin,\phi_\myout)$-clusters $C_1,\cdots, C_h$ for some $h\leq k$. Let $\vect{r}_i:=\frac{1}{\sqrt{|C_i|}}\cdot\1_{C_i}$. Then there exist $h$ orthonormal vectors $\tilde{\vect{r}}_1,\cdots, \tilde{\vect{r}}_h \in \Span(\vect{r}_1,\cdots,\vect{r}_h)$ and a constant $c_{\ref{thm:PSZ_structure}}>0$, such that 
	$$\norm{\vect{v}_i - \tilde{\vect{r}}_i}_2^2\leq c_{\ref{thm:PSZ_structure}}\cdot \frac{h\phi_\myout}{\phi_\myin^2}.$$
\end{theorem}
\begin{proof}
	Let $\rho(h):=\min_{A_1,\cdots,A_h}\max\{\phi_G(A_i): i=1,\cdots, h\}$, where the minimum is taken over all $h$-partitions $A_1,\cdots, A_h$. It is proven in Theorem 1.1 of \cite{PSZ17:partition} that if $\lambda_{h+1}/\rho(h)\geq c h^2$ for some constant $c>0$, then there exist orthonormal vectors $\tilde{\vect{r}}_1,\cdots,\tilde{\vect{r}}_h\in \Span(\vect{r}_1,\cdots,\vect{r}_h)$ such that
	$\norm{\vect{v}_i-\tilde{\vect{r}}_i}_2^2\leq 1.1 h\cdot \frac{\rho(h)}{\lambda_{h+1}}.$
	
	Note that by definition, $\rho(h)\leq \phi_\myout$. In addition, by Lemma~\ref{lemma:spectral_clusterable}, it holds that $\lambda_{h+1}\geq \frac{\phi_\myin^2}{2}$. %
	Furthermore, since $\phi_\myout\leq a_{\ref{thm:PSZ_structure}}\phi_\myin^2/k^2\leq a_{\ref{thm:PSZ_structure}}\phi_\myin^2/{h^2}$, it holds that $\lambda_{h+1}/\rho(h)\geq \frac{\phi_\myin^2}{2\phi_\myout}=ch^2$ as $a_{\ref{thm:PSZ_structure}}$ is sufficiently small constant. This then implies that  $\norm{\vect{v}_i-\tilde{\vect{r}}_i}_2^2\leq 1.1 h\cdot \frac{\rho(h)}{\lambda_{h+1}}\leq c_{\ref{thm:PSZ_structure}}\cdot \frac{h\phi_\myout}{\phi_\myin^2}$ for some constant $c_{\ref{thm:PSZ_structure}}$.
\end{proof}

Now we are ready to prove Theorem~\ref{thm:rw_clusterable}. We first provide a high level idea. We will bound the $\ell_2$-norm distance of the random walk distribution $\p_v^t$ and the uniform distribution $\uni_C$ over the cluster $C$ that contains $v$, i.e., $\norm{\p_v^t-\uni_C}_2$. In order to do so, we note that %
by Theorem~\ref{thm:PSZ_structure}, the vector $\uni_C$, which is a scale of the indicator vector of $C$, lies in a space that can be well approximated by the space of the first $h$ (where $h\leq k$ is the number of clusters) eigenvectors of matrix $\mat{P}$. Using this, we show that the projection of $\p_v^t-\uni_C$ on the space spanned by the first $h$ eigenvectors is small. Furthermore, by Lemma~\ref{lemma:spectral_clusterable}, $\lambda_{h+1}$ is large, and thus the length of the projection of $\p_v^t-\uni_C$ on the space spanned by the remaining $n-h$ eigenvectors is dominated by $(1-\frac{\lambda_{h+1}}{2})^{O(t)}$, which is also small for appropriately chosen $t$. Now we give the details.%
\begin{proof}[Proof of Theorem~\ref{thm:rw_clusterable}]
	For any vertex $v$, we let $X_v:=\sum_{j=1}^h\vect{v}_j(v)^2$. We first note that $\sum_{v\in V}X_v = \sum_{v\in V}\sum_{j=1}^h\vect{v}_j(v)^2=\sum_{j=1}^h\norm{\vect{v}_j}_2^2=h$. Therefore, by the averaging argument, there can be at most $\frac{\beta\alpha}{2}n$ vertices $v$ with $X_v\geq \frac{2h}{\beta\alpha n}$.
	
	Note that by the precondition of the Theorem, it holds that $\phi_\myout\leq a_{\ref{thm:PSZ_structure}}\phi_\myin^2/k^2$. Let $\vect{r}_i$ and $\tilde{\vect{r}}_i$ be the vectors as defined in Theorem~\ref{thm:PSZ_structure}. Let $Y_v:=\sum_{j=1}^h(\vect{v}_j(v) - \tilde{\vect{r}}_j(v))^2$. Then by applying Theorem \ref{thm:PSZ_structure} with graph $G$, we have that 
	\begin{eqnarray*}
		\sum_vY_v=\sum_{v}\sum_{j=1}^h(\vect{v}_j(v) - \tilde{\vect{r}}_j(v))^2 =\sum_{j=1}^h\norm{\vect{v}_j-\tilde{\vect{r}}_j}_2^2
		\leq c_{\ref{thm:PSZ_structure}}
		\cdot \frac{h^2\phi_\myout}{\phi_\myin^2}
	\end{eqnarray*}
	Again, by the averaging argument, there can be at most $\frac{\beta\alpha}{2}n$ vertices $v$ with $Y_v\geq c_{\ref{thm:PSZ_structure}}
	\cdot \frac{h^2\phi_\myout}{\phi_\myin^2}\frac{2}{\beta\alpha n}$.
	
	Now let us define $C'_i:=\{v: v\in C_i, X_v\leq\frac{2h}{\beta\alpha n}, Y_v\leq c_{\ref{thm:PSZ_structure}}
	\cdot \frac{h^2\phi_\myout}{\phi_\myin^2}\frac{2}{\beta\alpha n} \}$. Note that for any $C_i$ with $|C_i|\geq \alpha n$, it holds that $|C'_i|\geq |C_i|- (\frac{\beta\alpha}{2}+\frac{\beta\alpha}{2})n\geq |C_i|-\beta|C_i|\geq (1-\beta)|C_i|$. 
	
	Let us consider any vertex $v\in C_i'$. Since $\vect{r}_i=\frac{\1_{C_i}}{\sqrt{|C_i|}}$, it holds that  
	\begin{eqnarray*}
		\uni_{C_i}=\frac{\1_{C_i}}{|C_i|}=\agbracket{\1_v,\frac{\1_{C_i}}{\sqrt{|C_i|}}}\cdot\frac{\1_{C_i}}{\sqrt{|C_i|}}=\agbracket{\1_v,\vect{r}_i}\cdot \vect{r}_i=\sum_{j=1}^h \vect{r}_j(v) \cdot \vect{r}_j 
		& =& 
		\sum_{j=1}^h \tilde{\vect{r}}_j(v)\cdot \tilde{\vect{r}}_j，
	\end{eqnarray*}
	where the last equation follows from the fact that $\tilde{\vect{r}}_1,\cdots,\tilde{\vect{r}}_h$ have the same linear span as vectors $\vect{r}_1,\cdots,\vect{r}_h$, which in turn follows from the properties of $\{\tilde{\vect{r}}_i\}$ as guaranteed by Theorem~\ref{thm:PSZ_structure}.%

	Recall that $\p_v^{t}=\sum_{j=1}^n(1-\frac{\lambda_j}{2})^t\vect{v}_j(v)\cdot \vect{v}_j$. We let $t=\frac{20\log n}{\phi_\myin^2}$. Thus, we have that 
	\begin{eqnarray*}
		&&
		\norm{\p_v^{t}-\uni_{C_i}}_2 \\ 
		&= & \norm{\sum_{j=1}^n(1-\frac{\lambda_j}{2})^t\vect{v}_j(v)\cdot \vect{v}_j - \sum_{j=1}^h\tilde{\vect{r}}_j(v)\cdot \tilde{\vect{r}}_j}_2 \\
		& = &
		\norm{\sum_{j=1}^n(1-\frac{\lambda_j}{2})^t\vect{v}_j(v)\cdot \vect{v}_j - \sum_{j=1}^h\vect{v}_j(v)\cdot\vect{v}_j + \sum_{j=1}^h\vect{v}_j(v)\cdot\vect{v}_j - \sum_{j=1}^h\tilde{\vect{r}}_j(v)\cdot \tilde{\vect{r}}_j}_2 \\
		&\leq &
		\norm{\sum_{j=1}^h((1-\frac{\lambda_j}{2})^t-1)\vect{v}_j(v)\cdot\vect{v}_j}_2 + \norm{\sum_{j=h+1}^n(1-\frac{\lambda_j}{2})^t\vect{v}_j(v)\cdot\vect{v}_j}_2 + \norm{\sum_{j=1}^h\vect{v}_j(v)\cdot\vect{v}_j - \sum_{j=1}^h\tilde{\vect{r}}_j(v)\cdot \tilde{\vect{r}}_j}_2 \\
		&\leq&
		\sqrt{\sum_{j=1}^h((1-\frac{\lambda_j}{2})^t-1)^2\vect{v}_j(v)^2} + (1-\frac{\phi_\myin^2}{4})^t\sqrt{\sum_{j=h+1}^n\vect{v}_j(v)^2} + \norm{\sum_{j=1}^h\vect{v}_j(v)\cdot\vect{v}_j - \sum_{j=1}^h\tilde{\vect{r}}_j(v)\cdot \tilde{\vect{r}}_j}_2\\
		&\leq &
		(1-(1-\frac{\lambda_h}{2})^t) \sqrt{\sum_{j=1}^h\vect{v}_j(v)^2} + (1-\frac{\phi_\myin^2}{4})^t\sqrt{\sum_{j=h+1}^n\vect{v}_j(v)^2} + \norm{\sum_{j=1}^h\vect{v}_j(v)\cdot\vect{v}_j - \sum_{j=1}^h\tilde{\vect{r}}_j(v)\cdot \tilde{\vect{r}}_j}_2 \\
		&\leq & (1-(1-\phi_\myout)^t) \cdot \sqrt{X_v} + (1-\frac{\phi_\myin^2}{4})^t + \norm{\sum_{j=1}^h\vect{v}_j(v)\cdot\vect{v}_j - \sum_{j=1}^h\tilde{\vect{r}}_j(v)\cdot \tilde{\vect{r}}_j}_2 \\
		&& \qquad\qquad\qquad\qquad\qquad\qquad\qquad\qquad\qquad\qquad \textrm{(by Lemma~\ref{lemma:spectral_clusterable} and Fact~\ref{fact:sum_1})}\\
		&\leq&
		t\phi_\myout\cdot \sqrt{X_v} + \frac{1}{n^3} + \norm{\sum_{j=1}^h\vect{v}_j(v)\cdot\vect{v}_j - \sum_{j=1}^h\tilde{\vect{r}}_j(v)\cdot \tilde{\vect{r}}_j}_2 \quad \textrm{(by our setting $t= \frac{20\log n}{\phi_\myin^2}$)}\\
	\end{eqnarray*}

	Now observe that 
	\begin{eqnarray*}
		&&\norm{\sum_{j=1}^h\vect{v}_j(v)\cdot\vect{v}_j - \sum_{j=1}^h\tilde{\vect{r}}_j(v)\cdot \tilde{\vect{r}}_j}_2 \\
		&=&
		\norm{\sum_{j=1}^h\vect{v}_j(v)\cdot\vect{v}_j - \sum_{j=1}^h\vect{v}_j(v)\cdot \tilde{\vect{r}}_j +
			\sum_{j=1}^h\vect{v}_j(v)\cdot \tilde{\vect{r}}_j - \sum_{j=1}^h\tilde{\vect{r}}_j(v)\cdot \tilde{\vect{r}}_j}_2 \\
		&\leq &
		\norm{\sum_{j=1}^h\vect{v}_j(v)\cdot\vect{v}_j - \sum_{j=1}^h\vect{v}_j(v)\cdot \tilde{\vect{r}}_j}_2 + 
		\norm{\sum_{j=1}^h\vect{v}_j(v)\cdot \tilde{\vect{r}}_j - \sum_{j=1}^h\tilde{\vect{r}}_j(v)\cdot \tilde{\vect{r}}_j}_2\\
		&\leq&
		\sum_{j=1}^h\norm{\vect{v}_j(v)\cdot (\vect{v}_j -\tilde{\vect{r}}_j)}_2 + 
		\norm{\sum_{j=1}^h(\vect{v}_j(v) - \tilde{\vect{r}}_j(v))\cdot \tilde{\vect{r}}_j}_2 \\
		&\leq&
		\sqrt{c_{\ref{thm:PSZ_structure}}\cdot \frac{h\phi_\myout}{\phi_\myin^2}} \cdot \sum_{j=1}^h \abs{\vect{v}_j(v)} + \norm{\sum_{j=1}^h(\vect{v}_j(v) - \tilde{\vect{r}}_j(v))\cdot \tilde{\vect{r}}_j}_2  \quad \textrm{(by Theorem~\ref{thm:PSZ_structure})}\\
		&\leq & \sqrt{c_{\ref{thm:PSZ_structure}}\cdot \frac{h^2\phi_\myout}{\phi_\myin^2}} \cdot\sqrt{\sum_{j=1}^h\vect{v}_j(v)^2} + \sqrt{\sum_{j=1}^h(\vect{v}_j(v) - \tilde{\vect{r}}_j(v))^2} 
		\quad \textrm{(by Cauchy-Schwarz inequality)} \\
		&=&
		\sqrt{c_{\ref{thm:PSZ_structure}}\cdot \frac{h^2\phi_\myout}{\phi_\myin^2}} \cdot \sqrt{X_v}+ \sqrt{Y_v}
	\end{eqnarray*}

	Therefore,
	\begin{eqnarray*}
		\norm{\p_v^{t}-\uni_{C_i}}_2 
		&\leq& t\phi_\myout\cdot \sqrt{X_v} + \frac{1}{n^3} + \sqrt{c_{\ref{thm:PSZ_structure}}\cdot \frac{h^2\phi_\myout}{\phi_\myin^2}} \cdot \sqrt{X_v}+ \sqrt{Y_v}
		\\ 
		&\leq&
		\left(t\phi_\myout + \sqrt{c_{\ref{thm:PSZ_structure}}\cdot \frac{h^2\phi_\myout}{\phi_\myin^2}}\right) \cdot \sqrt{\frac{2h}{\beta\alpha n}} + \sqrt{c_{\ref{thm:PSZ_structure}}
			\cdot \frac{h^2\phi_\myout}{\phi_\myin^2}\frac{2}{\beta\alpha n}} +\frac{1}{n^3} 
		\leq  \frac{2\xi}{\sqrt{n}},
	\end{eqnarray*}
	where the last inequality follows from our setting that $h\leq k, t= \frac{20\log n}{\phi_\myin^2}$ and $\phi_\myout\leq a_{\ref{thm:rw_clusterable}}\frac{\xi{\alpha\beta}\phi_\myin^2}{k^3\log n}$, where $a_{\ref{thm:rw_clusterable}}>0$ is some sufficiently small constant. 
	
	Therefore, it holds that $\norm{\p_v^{t}-\uni_{C_i}}_{\TV}=\frac{1}{2}\norm{\p_v^{t}-\uni_{C_i}}_1\leq \frac{1}{2}\sqrt{n}\cdot \norm{\p_v^{t}-\uni_{C_i}}_2\leq \xi$. 	
\end{proof}

\subsection{From Clusterable Graphs to Noisy Clusterable Graphs}\label{subsec:noisywalks}
Now we analyze the random walk on a noisy clusterable graph $G$, for which we use an induced Markov chain introduced in~\cite{KPS13:noise} and some property of stopping rules of Markov chains~\cite{LW97:mixing}.

\paragraph{A tool: stopping rules of Markov Chains.} Consider a finite, irreducible, discrete time Markov chain on the state space $V=[n]$ with stationary distribution $\pi$. For any distribution $\sigma$, we let $\sigma^t$ denote the distribution of a $t$-step walk on the Markov chain with initial distribution $\sigma$. A \emph{stopping rule} $\Gamma$ of the Markov chain is a rule that observes the walk and decides whether to stop or not on the basis of what has been observed so far (see e.g.,~\cite{LW97:mixing} for formal definition). Given a starting distribution $\sigma$ and a target distribution $\tau$, we say that a stopping rule $\Gamma$ is a stopping rule from $\sigma$ to $\tau$ if the initial state is drawn from $\sigma$ and the final state is governed by $\tau$. Let $\E[\Gamma]$ denote the expected length before $\Gamma$ halts. For any two distributions $\sigma$ and $\tau$, we let $\HH(\sigma,\tau)$ denote the minimal expected length $\E[\Gamma]$ among all stopping rules $\Gamma$ from $\sigma$ to $\tau$. 

Let $\sigma^{(t)}$ denotes the distribution of a uniform average walk of length $t$ with initial distribution $\sigma$. The following lemma was proved by Lov{\'a}sz and Winkler.
\begin{lemma}[\cite{LW97:mixing}]\label{lemma:averagewalk}
	For any distribution $\tau$, and any subset $U\subset V$,  
	\begin{eqnarray*}
		\sum_{i\in U}\sigma^{(t)}(i) \leq \frac{1}{t}\HH(\sigma,\tau) + \frac{1}{t}\sum_{i\in U}\sum_{m=0}^{t-1}\tau^m(i)
	\end{eqnarray*}
	where $\tau^m$ denotes the probability vector of an $m$ step random walk on the Markov chain with initial distribution $\tau$. 
\end{lemma}
We remark that the above inequality was not explicitly stated in~\cite{LW97:mixing}, while the proof of Lemma 4.22 in \cite{LW97:mixing} directly implies the above Lemma.

\paragraph{An induced Markov chain.}
Let $G=(V,E)$ be a $d$-bounded graph. Let $\MM$ be the Markov chain corresponding to the (normal) random walks on the input graph $G$. For simplicity, we assume $\MM$ is irreducible (i.e., the graph is connected). By definition, the stationary distribution $\pi$ of $\MM$ is the uniform distribution $\uni_V$ on $V$, that is $\pi(i)=\frac{1}{n}$. %
Let $D$ denote a (large) subset of $V$ and let $B=V\setminus D$. Now we describe the new Markov chain $\MM'$, that has been considered in~\cite{KPS13:noise}, with state set $D$ as follows. For any two vertices $u,v\in D$, the transition probability $\p'_u(v)$ in $\MM'$ is the sum of $\p_u(v)$, i.e., the transition probability from $u$ to $v$ in $\MM$, and the probability $\bb_u^{(t)}(v)$ that is equal to the total probability of all length $t$ walks from $u$ to $v$ all of whose states, except for the end points $u$ and $v$ are in $B$, for any integer $t\geq 2$. That is, $\p'_u(v)=\p_u(v)+\sum_{t\geq 2}\bb_u^{(t)}(v)$. The chain $\MM'$ is formally constructed by first retaining the original transition in $\MM$ between $u,v$ and then adding new transitions $e_u^{(t)}(v)$ with transition probability $\bb_u^{(t)}(v)$ for any $t\geq 2$, for any $u,v\in D$. 

We note that the chain $\MM'$ is the \emph{stochastic complement} of $\MM$ with respect to set $D$~\cite{Mey89:stochastic}. Let $\mat{P}=\bigl(\begin{smallmatrix}
\mat{P}_{D}&\mat{P}_1 \\ \mat{P}_2&\mat{P}_B
\end{smallmatrix} \bigr)$ denote the transition probability matrix underlying $\MM$. We have the following lemma regarding the transition probability matrix $\mat{P}'$ underlying $\MM'$. 
\begin{lemma}[\cite{Mey89:stochastic}]\label{lemma:meyer}
	The Markov chain $\MM'$ is irreducible and aperiodic. Furthermore, its transition probability matrix is 
	$\mat{P}'= \mat{P}_D + \mat{P}_1(\mat{I}-\mat{P}_B)^{-1}\mat{P}_2$. 
\end{lemma}
It is known (see e.g., \cite{Mey89:stochastic} and \cite{KPS13:noise}) that, the stationary distribution in $\MM'$ is given by the vector $\pi'\in \R^{D}$ such that $\pi'(u)=\frac{\pi(u)}{\pi(D)}=\frac{1}{|D|}$ for any $u\in D$. 

Now let us consider a vertex $s\in D$ and an integer $\ell$ that will be specified later. Let $\tau:={\p'}_s^{(\ell)}$ denote the distribution of a random walk of length $\ell$ starting from $s\in D$ in $\MM'$. Consider the stopping rule $\Gamma$ that stops the walk in $\MM$ as soon as it has taken $\ell$ steps in $\MM'$, that is, $\Gamma$ is a stopping rule from $\1_s$ to $\tau$. Recall that $\E[\Gamma]$ denotes the expected number of steps the walk takes starting from $s$ before being terminated by the stopping rule $\Gamma$. The following lemma has been proven in~\cite{KPS13:noise}.

\begin{lemma}[\cite{KPS13:noise}]\label{lemma:hittingtime}
	There exists a set $\tilde{B}\subseteq D$ with $\pi(\tilde{B})\leq \pi(B)$ such that for any $s\in D\setminus \tilde{B}$, $\E[\Gamma]\leq 2\ell$. In particular, for any such vertex $s$, $\HH(\1_s,\tau)\leq 2\ell$.
\end{lemma}

Now we use the above induced chain to analyze the random walks on noisy clusterable graphs. Let $G$ be a graph with an $h$-partition $C_i$, $i\leq h$ satisfying the precondition of Theorem~\ref{thm:rw_perturbed}. We let $D$ denote the union of all $D_i$'s with $|D_i|\geq 2|B_i|$, that is, $D=\cup_{i:|D_i|\geq 2|B_i|} D_i$ and $B=V\setminus D$. We consider the induced Markov chain $\MM'$ with state set $D$.

Recall that we let $\mat{A}$ denote the adjacency matrix of the $d$-regular graph $G'$ corresponding to $G$ (see Section~\ref{sec:preliminaries}.) Then the transition probability matrix is $\mat{P}=\frac{\mat{I}+\frac{1}{d}\mat{A}}{2}$. If we let $\mat{A}=\bigl(\begin{smallmatrix}
\mat{A}_{D}&\mat{A}_1 \\ \mat{A}_2&\mat{A}_B
\end{smallmatrix} \bigr)$, then by Lemma~\ref{lemma:meyer}, the transition probability matrix of $G_{\mat{M}'}$ is 
\begin{eqnarray}
\mat{P}'=\frac{\mat{I}+\frac{1}{d}\mat{A}_D}{2}+\frac{\mat{A}_1}{2d}\left(\frac{\mat{I}-\frac{1}{d}\mat{A}_B}{2}\right)^{-1}\frac{\mat{A}_2}{2d}=\frac{\mat{I}+\frac{1}{d}(\mat{A}_D + \mat{A}_1(2d\mat{I}-\mat{A}_B)^{-1}\mat{A}_2)}{2}.\label{eqn:tranprobnew}
\end{eqnarray} 
If we let $G_{\MM'}$ denote the (weighted) $d$-bounded graph with adjacency matrix $\mat{A}_D + \mat{A}_1(2d\mat{I}-\mat{A}_B)^{-1}\mat{A}_2$, then by the above analysis (and the fact that $(2d\mat{I}-\mat{A}_B)^{-1}\geq \mat{0}$~\cite{Mey89:stochastic}), $\MM'$ corresponds to the lazy random walk on the graph $G_{\MM'}$. 

In the following, we show that $G_{\MM'}$ is a clusterable graph with clusters $D_i\subseteq D$, which will imply that the chain $\MM'$ has the nice local mixing property as guaranteed by Theorem~\ref{thm:rw_clusterable}. Then we can use the stopping rules to relate the chains $\MM'$ and $\MM$. 

The following lemma shows that if we construct $\MM'$ as above for the graph that satisfies the precondition of Theorem~\ref{thm:rw_perturbed}, then $G_{\MM'}$ is $(k,\phi_\myin,O(\phi_\myout))$-clusterable. This is trivial for the case of $k=1$ (as in~\cite{KPS13:noise}), as the inner conductance of any set is monotonically increasing. However, for general $k\geq 2$, we need to deal with the difficulty of bounding the outer conductance of potential clusters, as the outer conductance of any set is also monotonically increasing due to our construction.%
\begin{lemma}\label{lemma:inducedchain}
	Let $G=(V,E)$ be a $d$-bound graph with an $h$-partition $C_i$, $i\leq h$ such that $\phi_G(C_i)\leq \phi_\myout$. Furthermore, each $C_i$ can be partitioned into two subsets $D_i$ and $B_i$ such that $\phi(G[D_i])\geq \phi_\myin$. Let $D=\cup_{i:|D_i|\geq 2|B_i|} D_i$ and $B=V\setminus D$. Let $G_{\MM'}$ be the weighted graph corresponding to the Markov chain $\MM'$ on $D$ constructed as above. Then in the graph $G_{\MM'}$, each $D_i\subseteq D$ has the inner conductance at least $\phi_\myin$ and outer conductance at most $3\phi_\myout$.
\end{lemma}

\begin{proof}%
	We first consider the inner conductance of $D_i$ in $G_{\MM'}$. Let $S\subseteq D_i$ with $|S|\leq \frac{|D_i|}{2}$. By the fact that the adjacency matrix of $G_{\MM'}$ is $\mat{A}_D + \mat{A}_1(2d\mat{I}-\mat{A}_B)^{-1}\mat{A}_2$, it holds that $|E_{G_{\MM'}}(S,D_i\setminus S)|\geq |E_{G}(S,D_i\setminus S)|\geq \phi_\myin d|S|$. This implies that the inner conductance of $D_i$ in $G_{\MM'}$ is at least $\phi_\myin$.
	
	To bound the outer conductance of $D_i$ in $G_{\MM'}$, we instead bound the outer conductance $\phi_{\MM'}(D_i)$ of $D_i$ in the Markov chain $\MM'$, which is defined to be $\phi_{\MM'}(D_i):=\frac{\sum_{u\in D_i,v\in D\setminus D_i}\pi'(u)\p'_u(v)}{\pi'(D_i)}$, where $\p_u'(v)$ denotes the transition probability from $u$ to $v$ in the Markov chain $\MM'$. Note that by our definitions, $\phi_{G_{\MM'}}(D_i)=2\phi_{\MM'}(D_i)$.

	Recall that $\pi'(u)=\frac{1}{|D|}$ and that the transition probability matrix of $\MM'$ is $\mat{P}'$ given by Equation~(\ref{eqn:tranprobnew}). Then we have that 
	\begin{eqnarray}
	&&\sum_{u\in D_i,v\in D\setminus D_i}\pi'(u) \p_u'(v) 
	= 
	\frac{1}{|D|}\sum_{u\in D_i,v\in D\setminus D_i} \1_u \cdot \mat{P}'\cdot \1_v^T \nonumber \\
	&=&
	\frac{1}{|D|}\sum_{u\in D_i,v\in D\setminus D_i} \1_u \cdot \left(\frac{\mat{I}+\frac{1}{d}(\mat{A}_D + \mat{A}_1(2d\mat{I}-\mat{A}_B)^{-1}\mat{A}_2)}{2}\right) \cdot \1_v^T \nonumber \\
	&=&
	\frac{1}{|D|}\sum_{u\in D_i,v\in D\setminus D_i} \1_u \cdot \left(\frac{1}{2d}\left(\mat{A}_D + \mat{A}_1(2d\mat{I}-\mat{A}_B)^{-1}\mat{A}_2\right)\right) \cdot \1_v^T \nonumber \\
	&=& \frac{1}{|D|}\sum_{u\in D_i,v\in D\setminus D_i} \left( \frac{1}{2d} \left(\1_u \cdot \mat{A}_D\cdot \1_v^T + \frac{1}{2d}\1_u\cdot \mat{A}_1\cdot \sum_{j=0}^\infty (\frac{1}{2d}\mat{A}_B)^j\cdot \mat{A}_2 \cdot \1_v^T \right)\right) \nonumber  \\
\label{ineq:MCconduct}
	\end{eqnarray}
	where last equation follows from the Neumann Series  $(\mat{I}-\frac{\mat{A}_B}{2d})^{-1}=\sum_{j=0}^\infty (\frac{\mat{A}_B}{2d})^j$.
	
	
	We bound each term in the right hand side of the above inequality as follows. First, we have that 
	\begin{eqnarray}
	\sum_{u\in D_i,v\in D\setminus D_i}\1_u \cdot \mat{A}_D\cdot \1_v^T\leq |E_G(D_i,D\setminus D_i)|. \label{ineq:first_term}
	\end{eqnarray}
	Furthermore, we observe that $\1_u \cdot\mat{A}_1\cdot \mat{A}_2\cdot\1_v^T$ is exactly the number of paths that start from $u$, then go to a vertex $w\in B$, and then move to $v$. Thus, 	
	\begin{eqnarray*}
	\frac{1}{2d}\sum_{u\in D_i,v\in D\setminus D_i}\1_u \cdot\mat{A}_1\mat{A}_2\cdot\1_v^T 
	&\leq& \sum_{w\in B} \frac{|E_G(D_i,w)||E_G(w,D\setminus D_i)|}{2d}\nonumber \\
	&\leq& \sum_{w\in B_i} \frac{|E_G(w,D\setminus D_i)|}{2} + \sum_{w\in B\setminus B_i} \frac{|E_G(D_i,w)|}{2} \nonumber \\
	&=& \frac{1}{2}(|E_G(B_i,D\setminus D_i)|+|E_G(D_i,B\setminus B_i)|) \nonumber \\
	&\leq&\frac12|E_G(C_i,V\setminus C_i)| 
	\end{eqnarray*}

Similarly, for each $j\geq 1$, $\1_u \cdot\mat{A}_1\cdot \mat{A}_B^j\cdot \mat{A}_2\cdot\1_v^T$ is exactly the number of paths that start from $u$, then go to a vertex $w_1\in B$, and move inside $B$ for the next $j$ steps until some vertex $w_2\in B$, and then move to $v$. We have that 
%
%
\begin{eqnarray}
&&\frac{1}{(2d)^{j+1}}\sum_{u\in D_i,v\in D\setminus D_i}\1_u\cdot\mat{A}_1\mat{A}_B^j\mat{A}_2 \cdot \1_v^T \nonumber\\
&\leq& \frac{1}{(2d)^{j+1}}\sum_{w_1\in B}|E_G(D_i,w_1)|\cdot \sum_{\substack{w_2: p=(v_0=w_1,\cdots,v_{j}=w_2), \\v_\ell\in B, (v_\ell,v_{\ell+1})\in E(G)}}|E_G(w_2,D\setminus D_i)| \nonumber \\
&\leq& \frac{1}{(2d)^{j+1}}(\sum_{w_1\in B\setminus B_i}|E_G(D_i,w_1)|\cdot d^{j+1} +\sum_{w_2\in B_i}|E_G(w_2,D\setminus D_i)|\cdot d^{j+1}) \nonumber \\
&=&\frac{1}{2^{j+1}}(|E_G(D_i,B\setminus B_i)|+|E_G(B_i,D\setminus D_i)|)\nonumber \\
&\leq&\frac{1}{2^{j+1}}|E_G(C_i,V\setminus C_i)|,\nonumber
\end{eqnarray}
where in the first inequality, the third summation is taken over all possible paths $p$ from $w_1$ to some vertex $w_2\in B$, such that the length of $p$ is $j$ and all vertices on $p$ belong to $B$; in the second inequality, we used the fact that the number of such paths $p$ is at most $d^j$ and each vertex has degree at most $d$.

Thus, 
\begin{eqnarray}
\sum_{j=0}^\infty\frac{1}{(2d)^{j+1}}\sum_{u\in D_i,v\in D\setminus D_i}\1_u\cdot\mat{A}_1\mat{A}_B^i\mat{A}_2 \cdot \1_v^T \leq \sum_{j=0}^\infty\frac{1}{2^{j+1}} |E_G(C_i,V\setminus C_i)|= |E_G(C_i,V\setminus C_i)|\label{ineq:third_term}
\end{eqnarray}

	By the above inequalities~(\ref{ineq:MCconduct}),(\ref{ineq:first_term}),
	(\ref{ineq:third_term}), we obtain that 
	\begin{eqnarray*}
		\sum_{u\in D_i,v\in D\setminus D_i}\pi'(u) \p_u'(v) &\leq& \frac{1}{2d |D|} \cdot (1+1)\cdot|E_G(C_i,V\setminus C_i)| 
		=  \frac{|E_G(C_i,V\setminus C_i)|}{d|C_i|} \cdot \frac{|C_i|}{|D_i|}\cdot \frac{|D_i|}{|D|}\\
		&\leq& \phi_{G}(C_i) \cdot \frac{3}{2}\cdot \pi'(D_i)
		\leq \frac{3}{2}\phi_\myout \pi'(D_i),
	\end{eqnarray*}
	where in the second to last inequality, we used the assumption that $|D_i|\geq 2|B_i|$, which gives that $|D_i|\geq \frac{2}{3}|C_i|$. 
	
	Therefore, $\phi_{G_{\MM'}}(D_i)=2\phi_{\MM'}(D_i)\leq 3\phi_\myout$.

\end{proof}

Now we are ready to prove Theorem~\ref{thm:rw_perturbed}.
\begin{proof}[Proof of Theorem~\ref{thm:rw_perturbed}]
	Let $D=\cup_{j:|D_j|\geq 2|B_j|} D_j$. Let $B=V\setminus D$. Then it holds that $|B|= \sum_{1\leq i\leq h}|B_i|+\sum_{i:|D_i|< 2|B_i|} |D_i|\leq 3 \sum_{1\leq i\leq h}|B_i| \leq 3\varepsilon n$, and $|D|\geq (1-3\varepsilon)n$. We consider the induced Markov chain $\MM'$ on $D$ as above. By Lemma~\ref{lemma:inducedchain}, the corresponding $d$-bounded weighted graph $G_{\MM'}$ is $(k,\phi_\myin,3\phi_\myout)$-clusterable. In particular, $\phi_{G_{\MM'}}(D_i)\leq 3\phi_\myout$ and $\phi(G_{\MM'}[D_i])\geq \phi_\myin$ for any $D_i\subset D$.
	
	Let $\ell$ be an integer that will be specified later. For any $s\in D$, we let $\tau_s:={\p'_s}^{(\ell)}$ being the probability distribution of an $\ell$ step random walk starting from $s$ in the induced Markov chain $\MM'$. Let $\Gamma_s$ be the stopping rule from $\1_s$ to $\tau_s$ which is obtained by stopping the random walk that starts at $s$ in $\MM$ as soon as it has taken $\ell$ steps in $\MM'$. Let $\tilde{B}\subseteq D$ be the set guaranteed by  Lemma~\ref{lemma:hittingtime} such that $|\tilde{B}|\leq |B|\leq 3\varepsilon n$ and for any $s\in D\setminus \tilde{B}$, 
	\begin{eqnarray}
	E[\Gamma_s]\leq 2\ell.\label{eqn:gamma}
	\end{eqnarray}

	Now we set $a_{\ref{thm:rw_perturbed}}=\frac{a_{\ref{thm:rw_clusterable}}}{120}$ and thus $\phi_\myout\leq \frac{a_{\ref{thm:rw_clusterable}}\varepsilon\ccss^4\phi_\myin^2}{120k^3\log n}$. We then apply  Theorem~\ref{thm:rw_clusterable} on $G_{\MM'}$ with $(\phi_\myin,3\phi_\myout)$-clusters $D_i$ and $\alpha=3\sqrt{\varepsilon},\beta = 3\sqrt{\varepsilon}$, $\xi=\frac{\ccss}{6}$, to obtain that for any $D_j$ with $|D_j|\geq 3\sqrt{\varepsilon}n\geq 3\sqrt{\varepsilon}|D|$,%
	there exists a set $D_j'$ with $|D_j'|\geq (1-3\sqrt{\varepsilon})|D_j|$ such that for any $s\in D_j'$ and $\ell=\frac{20\log n}{\phi_\myin^2}$, it holds that 
	$\norm{\tau_s-\uni_{D_j}}_{\TV}\leq \frac{\ccss}{6}$. 
	This implies that 
	\begin{eqnarray}
	\norm{\tau_s-\uni_{C_j}}_{\TV} &\leq& \norm{\tau_s-\uni_{D_j}}_{\TV} + \norm{\uni_{D_j}-\uni_{C_j}}_{\TV} 
	\leq \frac{\ccss}{6} + \frac{|C_j\setminus D_j|}{|C_j|}\nonumber\\
	&=& \frac{\ccss}{6} + \frac{|B_j|}{|C_j|}
	\leq \frac{\ccss}{6}+\frac{\varepsilon n}{3\sqrt{\varepsilon}n}=\frac{\ccss}{6} +\frac{\sqrt{\varepsilon}}{3} %
	\label{eqn:mixing}
	\end{eqnarray}
	
	Now we set $\widehat{D}_j:=D_j'\setminus \tilde{B}$. Then it is guaranteed that for any $j$ with $|D_j|\geq 3\sqrt{\varepsilon}n$,  $|\widehat{D}_j|\geq (1-3\sqrt{\varepsilon})|D_j|-3\varepsilon n \geq (1-4\sqrt{\varepsilon})|D_j|$. %
	Thus, for any $s\in \widehat{D}_j$, both inequalities~(\ref{eqn:gamma}) and~(\ref{eqn:mixing}) hold.
	
	Now let us consider an arbitrary $s\in \widehat{D}_j$. Let $\tau=\tau_s$ and $\sigma=\1_s$. %
	By the precondition of the Theorem, we have that $t=\frac{120\log n}{\ccss\phi_\myin^2}=\frac{6\ell}{\ccss}$. We further recall that $\pp_s^t$ denotes the distribution of a uniform average walk of length $t$ with initial distribution $\sigma$ in the original chain $\MM$. %
	By applying Lemma~\ref{lemma:averagewalk} with $\sigma^{(t)}=\pp_s^t$ and distribution $\tau$, we obtain that for any $U\subset V$,
	\begin{eqnarray*}
		\sum_{i\in U}\pp_s^t(i) \leq \frac{1}{t}\HH(\sigma,\tau) + \frac{1}{t}\sum_{i\in U}\sum_{m=0}^{t-1}\tau^m(i),
	\end{eqnarray*}
	where $\tau^m$ denotes the distribution of an $m$ step random walk on $G$ with initial distribution $\tau$, that is $\tau^m=\tau\mat{P}^m$. (Here we slightly abuse the notation $\tau$ and use it to denote the distribution on $V$ by adding zero coordinates corresponding to vertices in $V\setminus D$). This further implies that for any set $C_j$ and any $U\subseteq V$,
	\begin{eqnarray*}
		\sum_{i\in U}(\pp_s^t(i)-\uni_{C_j}(i)) \leq \frac{1}{t}\HH(\sigma,\tau) + \frac{1}{t}\sum_{i\in U}\sum_{m=0}^{t-1}(\tau^m(i)-\uni_{C_j}(i))
	\end{eqnarray*}
	
	Therefore,
	\begin{eqnarray}
	\norm{\pp_s^t - \uni_{C_j}}_{\TV} 
	&\leq&
	\frac{1}{t}\HH(\sigma,\tau) + \frac{1}{t}\sum_{m=0}^{t-1}\norm{\tau^m-\uni_{C_j}}_{\TV} 
	\leq
	\frac{2\ell}{t} +\frac{1}{t}\sum_{m=0}^{t-1}\norm{\tau^m-\uni_{C_j}}_{\TV},\label{eqn:final_bound}
	\end{eqnarray}
	where the last inequality follows from inequality~(\ref{eqn:gamma}). Now recall that $\mat{P}=\frac{\I+d^{-1}\A}{2}$ denotes the transition probability matrix of the random walk. We will show the following claim.%
	\begin{claim}\label{claim:close}
		For any $0\leq m\leq t-1$, it holds that 
		$\norm{\uni_{C_j}\mat{P}^m - \uni_{C_j}}_{\TV}\leq \frac{\ccss}{3}.$
	\end{claim}
	Assuming that the above claim holds, we have that for any $0\leq m\leq t-1$, 
	\begin{eqnarray*}
		&&\norm{\tau^m -\uni_{C_j}}_{\TV} = \norm{\tau \mat{P}^m -\uni_{C_j}}_{\TV} 
		\leq \norm{\tau \mat{P}^m -\uni_{C_j}\mat{P}^m + \uni_{C_j}\mat{P}^m - \uni_{C_j}}_{\TV}\\
		&\leq& \norm{\tau \mat{P}^m -\uni_{C_j}\mat{P}^m}_{\TV} + \norm{\uni_{C_j}\mat{P}^m - \uni_{C_j}}_{\TV}
		\leq \norm{\tau -\uni_{C_j}}_{\TV} + \norm{\uni_{C_j}\mat{P}^m - \uni_{C_j}}_{\TV} \\
		&\leq&\frac{\ccss}{6} +\frac{\sqrt{\varepsilon}}{3}+ \frac{\ccss}{3} =\frac{\ccss}{2}+\frac{\sqrt{\varepsilon}}{3},%
	\end{eqnarray*}
	where the last inequality follows from Ineq.~(\ref{eqn:mixing}) and Claim~\ref{claim:close}. This, together with inequality~(\ref{eqn:final_bound}), gives that 
	\begin{eqnarray*}
		\norm{\pp_s^t - \uni_{C_j}}_{\TV} 
		\leq
		\frac{2\ell}{t} +\frac{1}{t}\cdot t\cdot (\frac{\ccss}{2}+\frac{\sqrt{\varepsilon}}{3})
		\leq 
		\frac{\ccss}{3} +\frac{\ccss}{2}+\frac{\sqrt{\varepsilon}}{3}%
		< \ccss + \sqrt{\varepsilon}.
	\end{eqnarray*}
	This will then finish the proof of the theorem. %
	
	Now we give the proof of Claim~\ref{claim:close}.
	\begin{proof}[Proof of Claim~\ref{claim:close}]
		For notational simplicity, we let $C=C_j$. We write $\mat{P}=\sum_{i=1}^n\eta_i\vv_i\vv_i^T$, where $\eta_i:=1-\frac{\lambda_i}{2}$ and $\vv_i$ ($1\leq i\leq n$) denote the $i$-th eigenvalue of $\mat{P}$, respectively. Let $\uni_C=\sum_{i}\alpha_i\vv_i$. Note that $\sum_{i=1}^{n}\alpha_i^2=\norm{\uni_C}_2^2=\frac{1}{|C|}$. 
		
		Note that
		$$\frac{\1_C}{|C|}\cdot (\I-\mat{P})\1_C^T = \frac{\1_C(d\I-\A)\1_C^T}{2d|C|}=\frac{\sum_{u\sim v}(\1_C(u)-\1_C(v))^2}{2d|C|}=\frac{\phi_G(C)}{2}\leq\frac{ \phi_\myout}{2},$$
		which gives that $1-|C|\cdot \uni_C \mat{P} \uni_C^T\leq \frac{\phi_\myout}{2}$. Thus, $1-|C|\sum_i\eta_i\alpha_i^2\leq \frac{\phi_\myout}{2}$, or equivalently, 
		$\sum_i\eta_i\alpha_i^2\geq \frac{1-\phi_\myout/2}{|C|}.$

		Let $H=\{i: \eta_i\geq 1-\frac{x\phi_\myout}{2} \}$, where $x=\frac{8}{\ccss^2}$. Then we have that 
		$\sum_{i\in H}\alpha_i^2 + (1-\frac{x\phi_\myout}{2})\sum_{i\notin H} \alpha_i^2\geq \frac{1-\phi_\myout/2}{|C|}.$	
		Thus, 
		$\sum_{i\in H}\alpha_i^2 + (1-\frac{x\phi_\myout}{2})(\frac{1}{|C|}-\sum_{i\in H}\alpha_i^2)\geq \frac{1-\phi_\myout/2}{|C|},$
		which gives that 
		$$\sum_{i\in H}\alpha_i^2 \geq \frac{x-1}{x\cdot |C|}, \quad \sum_{i\notin H}\alpha_i^2 \leq\frac{1}{x|C|}.$$
		
		Now we have that 
		\begin{eqnarray*}
			&& \norm{\uni_C \mat{P}^m-\uni_C}_2^2
			=
			\sum_i(\alpha_i\eta_i^m-\alpha_i)^2
			= \sum_i\alpha_i^2(1-\eta_i^m)^2
			\leq
			\sum_{i\in H} (1-(1-\frac{x\phi_\myout}{2})^m)^2\alpha_i^2 +\sum_{i\notin H}\alpha_i^2 \\
			&\leq& 
			\sum_{i\in H} (\frac{xt\phi_\myout}{2})^2 \alpha_i^2 + \frac{1}{x|C|}
			\leq (\frac{x^2t^2\phi_\myout^2}{4}+ \frac{1}{x})\frac{1}{|C|}
			< \frac{\ccss^2}{4|C|}, 
		\end{eqnarray*}
		where we used our choice of parameters which satisfy that $t\phi_\myout\leq \ccss^3/16$ and $x=\frac{8}{\ccss^2}$.
		
		On the other hand, if we let $\D_C$ denote the diagonal matrix such that $\D_C(u,u)=1$ if $u\in C$ and $0$ otherwise, then by Proposition 2.5 in~\cite{ST13:local}, it holds that for any $m\geq 0$, 
		$$\uni_C (\mat{P} \D_C)^m\1_C^T=\uni_C (\mat{P} \D_C)^m\1_V^T\geq 1-\frac{m\phi_G(C)}{2}\geq 1-\frac{m\phi_\myout}{2}.$$ 
		This gives that 
		$$\uni_C \mat{P} ^m\1_{V\setminus C}^T = 1 - \uni_C \mat{P} ^m\1_{C}^T\leq 1 - \uni_C (\mat{P} \D_C)^m\1_C^T \leq \frac{m\phi_\myout}{2}.$$
		
		Finally, by the above calculations, we have that
		\begin{eqnarray*}
			&&\norm{\uni_C\mat{P}^m-\uni_C}_{\TV} 
			= \frac{1}{2}\norm{\uni_C\mat{P}^m -\uni_C}_1
			\leq\frac{1}{2}(\uni_C\mat{P}^m\1_{V\setminus C}^T+\sum_{i\in C}|\uni_C\mat{P}^m(i) -\uni_C(i)|) \\
			&\leq& \frac{1}{2} (\frac{m\phi_\myout}{2} + \sqrt{|C|}\cdot \sqrt{\sum_{i\in C}(\uni_C\mat{P}^m(i) -\uni_C(i))^2})
			\leq \frac{1}{2} (\frac{t\phi_\myout}{2} + \sqrt{|C|}\cdot \norm{\uni_C\mat{P}^m-\uni_C}_2)\\
			&\leq&  \frac{\ccss^3}{64} + \frac{\ccss}{4}< \frac{\ccss}{3}.
		\end{eqnarray*}
		This finishes the proof of the Claim.
	\end{proof}	
	This finishes the proof of Theorem~\ref{thm:rw_perturbed}.
\end{proof}

\section{Conclusions}\label{sec:conclusions}
We gave the first robust clustering oracle and local filter for reconstructing the cluster structure of bounded degree graphs. Both algorithms run in sublinear times. %
To design and analyze our algorithms, we formalized and proved a new behavior of random walks in a noisy clusterable graph: a random walk of appropriately chosen length from a typical vertex in a large cluster of the clusterable part will mix well in the corresponding cluster, which might be of independent interest. 

It will be an interesting open question to design a local reconstruction algorithm that outputs a clusterable graph with better cluster-quality guarantee, especially to remove the $\Theta(\log n)$ gap between the inner conductances of the original graph and the corrected graph from our current result. In the property testing setting, such a gap was successfully closed, for both testing expansion (\cite{CS10:expansion} vs. \cite{KS11:expansion,NS10:expansion}) and for testing $k$-clusterability (\cite{CPS15:cluster} vs. \cite{CKKMP18:cluster}). However, for the local reconstruction setting, we even do not know how to remove such a logarithmic gap for reconstructing noisy expander graphs (i.e., $k=1$). As noted in~\cite{KPS13:noise}, for the case $k=1$, one already needs to have more refined definitions of strong/weak vertices and much stronger results about random walks in noisy expander graphs. Removing the logarithmic gap from our result for locally reconstructing cluster structure for general $k\geq 1$ can be as hard, if not harder. Similar question can be asked for removing the $\Theta(\log n)$ gap between the inner and outer conductance of the input instance of our robust clustering oracle. As we mentioned before, there is evidence in~\cite{CKKMP18:cluster} showing that this is difficult (for distribution distance based algorithms). 

\vspace{1em}

\paragraph{Acknowledgements.}

We are thankful to anonymous reviewers of FOCS 2018 and STOC 2019 for valuable comments.

\bibliographystyle{alpha}
\bibliography{clustering}

\appendix
\begin{center}\huge\bf Appendix \end{center}

\section{Proof of Lemma~\ref{lemma:strongvertices}}\label{sec:proof_rco}
\begin{proof}[Proof of Lemma~\ref{lemma:strongvertices}]
	First, note that if there are more than $\sqrt{\cst} |C|$ vertices in $C$ satisfying that $\pp_u^t(v)\leq (1-\sqrt{\cst})/|C|$, then $\norm{\pp_u^t-\uni_{C}}_{TV}>\sqrt{\cst} |C|\cdot \sqrt{\cst}/|C|\geq \cst$, which contradicts to the fact that $u$ is strong with respect to $C$.
	
	Second, by the definition of the set $S_u^{\theta_0}$ and the fact that ${\theta_0}\leq 1/2$, there can be at most $2\sqrt{n}$ vertices in $V\setminus S_u^{\theta_0}$, and thus there are at least $(1-\sqrt{\cst})|C|-2\sqrt{n}$ vertices $w\in S_u^{\theta_0} \cap C$ such that $\pp_u^t(w)\geq \frac{1-\sqrt{\cst}}{|C|}$. Thus
	\begin{eqnarray*}
		\pp_u^t(S_u^{\theta_0}) \geq ((1-\sqrt{\cst})|C|-2\sqrt{n}) \cdot \frac{1-\sqrt{\cst}}{|C|}\geq 1/2,
	\end{eqnarray*}
	where in the second inequality we used the fact that $|C|\geq 3\sqrt{\varepsilon'} n = 3\sqrt{\frac{6\varepsilon}{\phi}}n > \frac{8\sqrt{n}}{\sqrt{\kappa}}$ as $\varepsilon=\Omega(\frac{\phi}{{n}})$.
	
	Finally, since $u$ is strong with respect to $C$, there are at least $(1-\sqrt{\cst})|C|-2\sqrt{n}$ vertices $w\in S_u^{\theta_0} \cap C$ such that $\pp_u^t(w)\geq \frac{1-\sqrt{\cst}}{|C|}$. The same is true for $v$. Thus, there are at least $(1-2\sqrt{\cst})|C|-4\sqrt{n}$ vertices $w\in S_u^{\theta_0}\cap S_v^{\theta_0}\cap C$ such that $p_u(w),p_v(w)\geq \frac{1-\sqrt{\cst}}{|C|}$. Again, by the fact that $|C|>\frac{8\sqrt{n}}{\sqrt{\kappa}}$, we have that
	\begin{eqnarray*}
		\rcp_{\theta_0}(u,v)
		&\geq& ((1-2\sqrt{\cst})|C|-4\sqrt{n})\cdot \frac{1-\sqrt{\cst}}{|C|} \cdot \frac{1-\sqrt{\cst}}{|C|}
		\geq
		\frac{1-5\sqrt{\cst}}{|C|}. %
	\end{eqnarray*}
This finishes the proof of the Lemma.
\end{proof}

\section{Description of the Algorithm \textsc{EstimateRCP}}\label{sec:RCP_alg}
In the algorithm, $C$ is a sufficiently large constant.
\begin{center}
	\begin{tabular}{|p{\textwidth}|}
		\hline
		\textsc{EstimateRCP}$(G, u,v, \theta, \delta, t)$ \\
		\hline
		\begin{tightenumerate}
			\item Run the following $C\log n$ times:
			\begin{tightenumerate}
				\item Let $F_u:=\textsc{FindSet}(G,u,\theta,t)$ and $F_v:=\textsc{FindSet}(G,v,\theta, t)$
				\item Keep performing uniform average walks of length $t$ from $u$ (resp. $v$) until $x:=\sqrt{n}/\delta^2$ such walks end at vertices in $F_u$ (resp. $F_v$). Let $W_u$ (resp. $W_v$) denote the set of walks. If more than $20x$ walks are performed (from either $u$ or $v$), then report \textsc{Fail}.
				\item Let $A$ be the number of pairwise collisions\footnote{If a walk from $W_u$ and a walk from $W_v$ end at the same vertex, then this counts as one pairwise collision.} between walks in $W_u$ and $W_v$. Output $A/x^2$. 
			\end{tightenumerate}				
			\item If the majority of the above runs do not fail, then output the median of all the  output numbers in successful runs. Otherwise, \textsc{Abort}.
		\end{tightenumerate}\\
		\hline
	\end{tabular}
\end{center}
\begin{center}
	\begin{tabular}{|p{\textwidth}|}
		\hline
		\textsc{FindSet}$(G,u,\theta, t)$ \\
		\hline
		\begin{tightenumerate}
			\item Perform $C\sqrt{n}\log n$ independent uniform average walks of length $t$ from $u$.
			\item Let $F_u$ denote the set of all vertices $w$ such that at most $C(1-\frac{\theta}{2})\log n$ walks from $u$ end at $w$. Return $F_u$.
			\vspace{-1em}
		\end{tightenumerate}\\
		\hline
	\end{tabular}
\end{center}

\section{Further Guarantees on the Locally Reconstructed Graph}\label{sec:cluster_outersmall}
In the following, we show that by sacrificing the inner conductance quality, we can also find a clustering of the reconstructed graph $G'$ with small outer conductance. 
\begin{lemma}\label{lemma:out_conductance}
	Let $\phi^*=\frac{a_{\ref{lemma:conductance}}\varepsilon\phi}{k^4\log n}$. %
	If $G$ is an $\varepsilon$-perturbation of a $(k,\phi,\frac{a_{\ref{thm:rw_perturbed}}\varepsilon\kappa^4\phi}{3k^3\log n})$-clusterable graph, then the resulting graph	$G'$ from the local reconstruction algorithm is $(k,\frac{\nu^6}{6^k}\phi^*,\min\{k\nu, 1\})$-clusterable, for any $0\leq \nu\leq 1$.
\end{lemma}
\begin{proof}
We start with the $(k,\phi^*,1)$-clustering of $G'$ that is guaranteed from Lemma~\ref{lemma:conductance}. Let $C_1,\cdots, C_h$ be a partition satisfying that $\phi(G'[C_i])\geq \phi^*$. 
Let $\nu\in [0,1]$. %
We next carefully merge some of these clusters so that each part of the final partition will have both inner conductance at least $\frac{\nu^k}{6^k}\phi^*$ and outer conductance at most $\min\{k\nu, 1\}$. 

If there exists $1\leq i\neq j\leq h$ such that $|C_i|\leq |C_j|$ with $|E'(C_i,C_j)|\geq \nu d|C_i|$, then we merge $C_i$ and $C_j$ to obtain a new cluster $C:=C_i\cup C_j$. We repeat until the condition is violated. 

Note that this process always terminates as each time the number of clusters decrease by $1$. Furthermore, note that after termination, each cluster has outer conductance at most $\min\{1,k\nu\}$ by construction. Now we show that in each iteration, the merged $C=C_i\cup C_j$ still has large inner conductance. Let $S\subset C$ with $|S|\leq \frac{|C|}{2}$. Let $S_i=S\cup C_i$ and $S_j= S\cup C_j$. Note that it can not happen simultaneously that $|S_i|>\frac{|C_i|}{2}$ and $|S_j|>\frac{|C_j|}{2}$. Now we have the following cases.
\begin{itemize}[leftmargin=15pt]
	\item If both $|S_i|\leq \frac{|C_i|}{2}$ and $|S_j|\leq \frac{|C_j|}{2}$, then $$\phi_{G[C]}(S)=\frac{|E'(S,C\setminus S)|}{d|S|}\geq \min\{\frac{|E'(S_i,C_i\setminus S_i)|}{d|S_i|}, \frac{|E'(S_j,C_j\setminus S_j)|}{d|S_j|}\}\geq \phi^*.$$
	\item If $|S_j|> \frac{|C_j|}{2}$, then $|S|\leq |C_i|+|S_j|\leq |C_j|+|S_j|<3|S_j|$. 
	\begin{tightenumerate}
		\item If $|S_j|\geq (1-\frac{\nu}{2})|C_j|$, then $|C_i|\geq \frac{2}{3}|C_j|$ as otherwise $|C|\leq \frac{5}{3}|C_j|$ and $|S|\geq |S_j|>\frac{|C|}{2}$, a contradiction. Then $|S_i|\leq \frac{\nu}{2}|C_j|\leq \frac{\nu}{2}\frac{3}{2}|C_i|=\frac{3\nu}{4}|C_i|$. Thus there will be at least $\frac{d\nu}{4}|C_i|$ edges between $S_j$ and $C_i\setminus S_i$. Thus $\phi_{G[C]}(S)\geq \frac{|E'(S_j,C_i)|}{d|S|}\geq \frac{\frac{d\nu}{4}|C_i|}{3d|S_j|}\geq \frac{\frac{d\nu}{4}\frac{2}{3}|C_j|}{3d|C_j|}=\frac{\nu}{18}$.
		\item If $|S_j|\leq (1-\frac{\nu}{2})|C_j|$, then $|C_j\setminus S_j|\geq \frac{\nu}{2}|C_j|\geq \frac{\nu}{2(1-\frac{\nu}{2})}|S_j|$. Therefore, $\phi_{G[C]}(S)\geq  \frac{|E'(S_j,C_j\setminus S_j)|}{d|S|}\geq \frac{\phi^* d|C_j\setminus S_j|}{3d|S_j|}>\frac{\phi^*\nu}{6}$.
	\end{tightenumerate}
	\item If $|S_i|>\frac{|C_i|}{2}$, then it must hold that $|S_j|< \frac{|C_j|}{2}$.
	\begin{tightenumerate}
		\item If $|S_i|<(1-\frac{\nu}{2})|C_i|$, then $\frac{|C_i|}{2}\geq |C_i\setminus S_i|\geq \frac{\nu}{2}|C_i|$. Thus  
		$\phi_{G[C]}(S)\geq \frac{|E'(S_i,C_i\setminus S_i)|+|E'(S_j,C_j\setminus S_j)|}{d(|S_i|+|S_j|)}\geq \min\{\frac{\phi^*d |C_i\setminus S_i|}{d|S_i|}, \frac{\phi^* d |S_j|}{d|S_j|} \} = \min\{\frac{\nu\phi^*}{2}, \phi^*\}=\frac{\nu\phi^*}{2}$.
		\item If $|S_i|\geq (1-\frac{\nu}{2})|C_i|$, then $|E'(S_i,C_j)|\geq \frac{d\nu}{2}|C_i|$. If $|E'(S_i,S_j)|\geq \frac{1}{2}|E'(S_i,C_j)|$, then $|S_j|\geq \frac{\nu}{4}|C_i|$, then $\phi_{G[C]}(S)\geq \frac{|E'(S_j,C_j\setminus S_j)|}{d|S|}\geq \frac{\phi^* d|S_j|}{d(|S_j|+|C_i|)}\geq \frac{\phi^* \nu}{5}$. Otherwise, $|E'(S_i,S_j)|< \frac12|E'(S_i,C_j)|$, then $|E'(S_i,C_j\setminus S_j)|\geq \frac12|E'(S_i,C_j)|\geq \frac{d\nu}{4}|C_i|$. Thus  $\phi_{G[C]}(S)\geq \frac{|E'(S_i,C_j\setminus S_j)|+|E'(S_j,C_j\setminus S_j)|}{d(|S_i|+|S_j|)}\geq \min\{\frac{\frac{d\nu}{4}|C_i|}{d|S_i|}, \frac{\phi^* d |S_j|}{d|S_j|}\}\geq \min\{\frac{\nu}{4}, \phi^* \}$.
	\end{tightenumerate}
\end{itemize}
From the above analysis, we know that if both $\phi(G[C_i])\geq \phi^*$ and $\phi(G[C_j])\geq\phi^*$, then after merging $C_i$ and $C_j$, the resulting cluster $C$ has inner conductance at least $\frac{\nu\phi^*}{6}$. Since there will be at most $k$ iterations (or merges), we know that in the final partition $\mathcal{P'}$, each part has outer conductance at most $\min\{k\nu,1\}$ and inner conductance $\frac{\nu^k\phi^*}{6^k}=\frac{a_{\ref{lemma:conductance}}\nu^k}{6^kk^4}\frac{\varepsilon\phi}{\log n}$. This proves the statement of the lemma.%
\end{proof}

\end{document}